\documentclass[11pt, a4paper]{article}
\usepackage{graphicx}
\usepackage[a4paper]{geometry}
\usepackage{amsmath}
\usepackage{amsthm}
\usepackage{amssymb}
\usepackage{amsfonts}
\usepackage{epsf}
\usepackage{graphicx}
\usepackage{epstopdf}
\usepackage{hyperref}
\usepackage{color}
\usepackage{cancel}
\usepackage{comment}
\usepackage{appendix}

\numberwithin{equation}{section}
\title{On the Fredholm determinant of  the confluent hypergeometric kernel with  discontinuities }
\author{Shuai-Xia Xu$^1$, Shu-Quan Zhao$^{2,*}$, and Yu-Qiu Zhao$^2$\\\\
$^1$Institut Franco-Chinois de l'Energie Nucl\'{e}aire, Sun Yat-sen University, \\Guangzhou 510275, China\\
$^2$Department of Mathematics, Sun Yat-sen University, Guangzhou 510275, China\\
$^*$E-mail address: zhaoshq25@mail2.sysu.edu.cn}

\begin{document}
\date{}
\maketitle

\maketitle

\noindent \hrule width 6.27in\vskip .3cm

\noindent {\bf{Abstract } }
We consider the determinantal point process with the confluent hypergeometric kernel. This process is a universal point process in random matrix theory and describes  the distribution of eigenvalues of large random Hermitian matrices  near the  Fisher-Hartwig singularity. Applying the Riemann-Hilbert method,
 we study the generating function of this process on any given number  of intervals. It can be expressed as the Fredholm determinant of the confluent hypergeometric kernel with $n$ discontinuities.  In this paper, we derive  an integral representation for the determinant by using the Hamiltonian of the  coupled Painlev\'e V system. By evaluating the total integral of the Hamiltonian, we obtain the asymptotics of the determinant as the $n$ discontinuities tend to infinity up to and including the constant term. Here the constant term is expressed in terms of  the Barnes  $G$-function.

  \vskip .5cm

\tableofcontents

\noindent

\section{Introduction}
In random matrix theory, one of the most important problems is to study the local     statistics of eigenvalues for large random matrices. Compared with the global perspective in which the limiting average densities of eigenvalues are diverse, the local statistics of eigenvalues for large  random matrices display remarkable universality. For instance, for general large unitary invariant random matrices, the local eigenvalue distributions near a regular point in the bulk, near the soft edge and the hard edge can be described by the determinantal point processes associated with the sine kernel, Airy kernel and Bessel kernel, respectively; see \cite{Ake1,For1,Meh1}.
The gap probabilities in these  processes can then be expressed in terms of the Fredholm determinants with the kernels associated. It is remarkable that the Fredholm determinants admit integral representations in terms of the Painlev\'e transcendents; \cite{Ake1, For1, For2, Jimbo1,Tracy1,Tracy2}.

 A different type of local statistics  will arise if we consider the unitary ensemble with Fisher-Hartwig singularity as follows \cite{ deift1,FM,KV}:
\begin{equation}\label{jointprodensity}
    p_n\left(x_1, \dots, x_n\right)=\frac{1}{Z_n} \prod_{i=1}^n e^{-V(x_i)}\left|x_i\right|^{2 \alpha} \chi_\beta\left(x_i\right) \prod_{i<j}\left(x_i-x_j\right)^2,
\end{equation}
where $Z_n$ is the normalization constant, the potential $V$ is a general real valued function with enough increase at infinity.
The parameters {satisfy} $\alpha>-\frac{1}{2}$,  $\beta \in i \mathbb{R}$, and
$$
\chi_\beta(x)=\left\{\begin{array}{ll}
e^{\beta \pi i}, & x<0, \\
e^{-\beta \pi i}, & x>0.
\end{array}  \right.
$$
{The root-type singularity of the form $|x-a|^{2\alpha}$ and the  jump-type singularity such as $\chi_\beta(x-a)$ are known as Fisher-Hartwig (FH) singularities in the literature. The asymptotics of the large Toeplitz determinants with these types of singularities were first investigated  and some formulas on the asymptotics were conjectured by Fisher and Hartwig in \cite{FHT1968} and subsequently in \cite{LS1972} by Lenard,  motivated by their applications  in Ising models, random walks and one-dimensional Bose gases.}

When the limiting average density function of the ensemble \eqref{jointprodensity} is positive at the origin, the {scaling}  limit of the  correlation function near the origin is now described by
the following confluent hypergeometric kernel
\begin{equation}\label{chfkernel}
K^{(\alpha, \beta)}(x, y)=\frac{1}{2 \pi i} \frac{\Gamma(1+\alpha+\beta) \Gamma(1+\alpha-\beta)}{\Gamma(1+2 \alpha)^2} \frac{\mathbb{A}(x) \mathbb{B}(y)-\mathbb{A}(y)\mathbb{ B}(x)}{x-y},
\end{equation}
where
\begin{equation}\label{eq:kernelA}
\mathbb{A}(x)=\chi_\beta(x)^{1 / 2}|2 x|^\alpha e^{-i x} \phi(1+\alpha+\beta, 1+2 \alpha, 2 i x),
\end{equation}
and $ \mathbb{B}(x)$ is the complex conjugation of $\mathbb{A}(x)$. Here $\phi(a,b,z)$ denotes the confluent hypergeometric function
\begin{equation}\label{eq:chf}
\phi(a, b, z)=1+\sum_{k=1}^{\infty} \frac{a(a+1) \dots(a+k-1) z^k}{b(b+1) \dots(b+k-1) k !}.
\end{equation}
For special parameters, the confluent hypergeometric kernel reduces to the sine kernel and Bessel kernel:
\begin{equation}\label{Sinekernel}
    K^{(0,0)}(x, y)=K_{\sin }(x, y)=\frac{\sin (x-y)}{\pi(x-y)},
\end{equation}
\begin{equation}\label{besselkernel}
    K^{(\alpha, 0)}(x, y)=K_{\text {Bess,1 }}(x, y)=\frac{|x|^\alpha|y|^\alpha}{x^\alpha y^\alpha} \frac{\sqrt{x y}}{2} \frac{J_{\alpha+\frac{1}{2}}(x) J_{\alpha-\frac{1}{2}}(y)-J_{\alpha-\frac{1}{2}}(x) J_{\alpha+\frac{1}{2}}(y)}{x-y}.
\end{equation}

Analogously to the sine, Airy and Bessel point processes, the determinantal point  process defined by the  confluent hypergeometric kernel represents another universality class in random matrix theory.  It describes  the local statistics of the eigenvalues of general random {matrices} near the Fisher-Hartwig {singularity} as shown in \eqref{jointprodensity}. This process arises also in the studies of the  representation theory of the infinite-dimensional unitary group \cite{borodin2002fredholm,BO}.

 In the present work, we consider
  the Fredholm determinant of the  confluent hypergeometric kernel \eqref{chfkernel} with discontinuities
   \begin{equation}\label{deformedfred1}
  F(t \vec{r}, \vec{\gamma};\alpha,\beta) = \det \left( I-K_{\sigma}^{(\alpha,\beta)}
 \right), \end{equation}
where  $t>0$, and $K_{\sigma}^{(\alpha,\beta)}$ denotes the integral operator acting on $L^2(tr_0, tr_n)$ with the  confluent hypergeometric kernel with discontinuities $\sigma(y) K^{(\alpha, \beta)}(x,y)$. Here the function $\sigma(x)$ is a step function defined by
\begin{equation}\label{sigma}
        \sigma(x)=\sum_{k=0}^{n-1} \gamma_k \chi_{\left(r_k t, r_{k+1} t\right)}(x) .
\end{equation}
The parameters  $\vec{\gamma}=\left(\gamma_0,  \dots, \gamma_{n-1}\right) \in[0,1]^n$, $\vec{r}=\left(r_0, r_1, \dots,
r_n\right) \in \mathbb{R}^{n+1}$  satisfy $r_0< r_1<\dots <r_n$ and $r_m=0$ for some $0\leq m\leq n$.
The determinant represents the  moment generating function of the confluent hypergeometric process.
To see this, we denote by $N(s)$ the random variable counting the number of particles in this process falling into the interval $[0,s]$ for $s>0$ and  $[s,0]$ for $s<0$.  The determinant \eqref{deformedfred1} can be written as the moment generating function
\begin{equation}\label{exponential moments}
    F(t \vec{r}, \vec{\gamma};\alpha,\beta)=\mathbb{E}\left(\prod_{j=0}^{n} e^{2 \pi i h_j N(t r_j)} \right),
\end{equation}
where $h_j$ is related to the parameters $\vec{\gamma}$ by
 $   h_j=\frac{1}{2 \pi i} \ln \frac{1-\gamma_{j-1}}{1-\gamma_j} $ for $m< j\leq n$  and  $   h_j=\frac{1}{2 \pi i} \ln \frac{1-\gamma_{j}} {1-\gamma_{j-1}} $ for $0\leq j\leq m$ with $\gamma_{-1}=\gamma_n=0$.

The  Fredholm determinant \eqref{deformedfred1}, or equivalently the moment generating function \eqref{exponential moments},
contains information about the statistics of the  confluent hypergeometric process.  For example,   by taking $n=2$ in \eqref{exponential moments} we have   \begin{equation}\label{gapPro1}
    F((r_0t,r_1t,r_2t),\gamma(1,1);\alpha,\beta)=\det \left(1-\gamma \left.K^{(\alpha,\beta)}\right|_{(r_0t,r_2t)}\right),
\end{equation}
where $\left.K^{(\alpha,\beta)}\right|_{(r_0t,r_2t)}$ denotes the integral operator acting on $L^2(r_0t, r_2t)$ with the  confluent hypergeometric kernel $ K^{(\alpha, \beta)}(x,y)$.
Therefore, the quantity $F((r_0t,r_1t,r_2t), (1,1);\alpha,\beta)$  gives the gap probability of finding no particles of the confluent hypergeometric process on the interval $(r_0t, r_2t)$.  {While} $F((r_0t, r_1t,r_2t), \gamma(1,1);\alpha,\beta)$  is  the gap probability of having no particles  falling in the interval $(r_0t, r_2t)$ in the thinned confluent hypergeometric process, where each of the particles in the process is removed independently with the probability $1-\gamma\in[0,1)$ \cite{Boho1}.  Generally,  $F((r_0t,0, r_2t), (\gamma_1,\gamma_2);\alpha,\beta)$ with $r_0<0$ and $r_2>0$,  describes  the probability to find no particles in the interval $(r_0t, r_2t)$ in the thinned confluent hypergeometric process, where the positive particles  in the process are removed independently with the probability $1-\gamma_2\in[0,1)$ and the  negative particles  in the process are removed independently with the probability $1-\gamma_1\in[0,1)$.

In \cite{deift2011asymptotics},    the large gap asymptotics of  the confluent hypergeometric process, which is given by
$F((-t,0,t), (1,1);\alpha,\beta)$, was derived by  Deift, Krasovsky, and Vasilevska by  considering the asymptotics of  the Toeplitz determinant generated by certain symbols with Fisher-Hartwig singularity on the unit circle.  Later in \cite{xu2020gap},  the first and third authors of the present paper established an integral representation of the Fredholm determinant $F((-t,0,t),(1,1);\alpha,\beta)$ in terms of the Hamiltonian of a coupled Painlev\'e V system by considering the double scaling limit of  the Toeplitz determinant with merging singularitites.  By evaluating the total integral of the Hamiltonian,  the large gap asymptotics was then derived  therein. Recently, the large gap probability asymptotics of the deformed Fredholm determinant $F((-t,0,t), (\gamma,\gamma);\alpha,\beta)$ was derived by Dai and Zhai in \cite{dai2022asymptotics} by evaluating again the total integral of the Hamiltonian of the coupled Painlev\'e V system. It is also worth mentioning  that  the  integral representations of the gap probabilities and their asymptotics of some other determinantal processes arising in random matrix theory have been considered, such as the Pearcey process and the hard edge Pearcey process in \cite{CM2023,DXZ1, DXZ2, DXZ3}, the higher order Airy processes \cite{CCG, DMS, X} and the process defined by the Painlev\'e II kernel \cite{bi,  XD}. The integral representations and the asymptotics of the moment generating  functions for the  sine,  Airy and Bessel point processes are considered in \cite{charlier2021exponential,charlier2021large,charlier2020large,charlier2019generating,Claseys2018The}.

{The studies of the asymptotics of the Fredholm determinants  with jump-type FH singularities similar to  \eqref{deformedfred1} are closely related to the problem of deriving the asymptotics of the Toeplitz and Hankel determinants with FH singularities. The problem has a long history tracing back to the pioneering work of Fisher and Hartwig \cite{FHT1968}, and has been widely explored due partially to  their various important applications of which we have mentioned a few in the previous context.  For example, the conjectures proposed in \cite{FHT1968} was proved later by Widom in \cite{WT1973}. The asymptotics of the Toeplitz determinants with several  separate FH  singularities  have been obtained by Deift, Its and Krasovsky in \cite{deift1, DIKO2014}. Recently, the asymptotics of the Toeplitz determinants with merging  singularities  have been  derived in \cite{cik, CKT2015, FU2021}. There are also  important progresses in the studies of the asymptotics of the Hankel determinants with FH singularities \cite{KC2007, ik, BCIH2016, CFR2016, CDA2018, WXZG2018, CA2019, CGA2021, CFWWA2021}.}

In the present work, we formulate a Riemann-Hilbert (RH, for short) representation for  the determinant \eqref{deformedfred1} of the confluent hypergeometric kernel \eqref{chfkernel}  with $n$ discontinuities by using the general theory of integrable operator developed in \cite{its1990differential}. By exploring the RH problem, we obtain
an integral representation for the determinant \eqref{deformedfred1} in terms of the Hamiltonian of the coupled Painlev\'e V system of $2n$ dimensions.  The large gap asymptotics of the  determinant up to and including the constant term are then derived by evaluating the integral of the Hamiltonian.

\subsection{Main results}

To state our main results, it is convenient to introduce the couple Painlev\'e V system defined by the following Hamiltonian dynamic system
\begin{equation}\label{hamiltonform}
    \frac{d v_k}{d t}=\frac{\partial H}{\partial u_k}, \quad \frac{d u_k}{d t}=-\frac{\partial H}{\partial v_k},\quad
    k=0,1,\dots,n,\quad  k\neq m
\end{equation}
for {any specific $m$ and} $0\leq m\leq n$, where the Hamiltonian $H$ is defined by
\begin{equation}\label{defhamilton}
    t {H(t)}=\sum_{k \neq m}-2 i t r_k H_V(u_k, v_k,-2 i t r_k ; \alpha, \beta)+\frac{1}{2} \sum_{j \neq k;\; j,k\neq m} u_j u_k(v_j+v_k)(v_j-1)(v_k-1),
\end{equation}
{with $r_k\neq0$ for $k\neq m$.} Here $H_V$  denotes  the  Hamiltonian for the Painlev\'e V equation defined in \cite{For3,Jimbo1981Monodromy}
\begin{equation}\label{eq:hamilton}
    {s H_V(u, v,s ; \alpha, \beta)}=-s u v-\alpha u\left(v^2-1\right)-\beta u(v-1)^2+u^2 v(v-1)^2.
\end{equation}
From  \eqref{hamiltonform}-\eqref{eq:hamilton}, the coupled Painlev\'e V system can be expressed explicitly as follows
\begin{equation}\label{eq:CPV}
    \left\{\begin{array}{l}
\begin{aligned}
    t \frac{d u_k}{d t}=&-2 i t u_k r_k-u_k \sum_{j \neq m} u_j\left(v_j-1\right)^2-2 u_k v_k\left(\sum_{j \neq m} u_j\left(v_j-1\right)\right)\\
    &+2(\alpha+\beta) u_k v_k-2 \beta u_k,
\end{aligned}
\\
\begin{aligned}
  t \frac{d v_k}{d t}=&2 i t v_k r_k+v_k\left(\sum_{j \neq m} u_j\left(v_j-1\right)^2\right)+v_k^2\left(\sum_{j \neq m} u_j\left(v_j-1\right)\right)\\
  &-\sum_{j \neq m} u_j v_j\left(v_j-1\right)-\alpha\left(v_k^2-1\right)-\beta\left(v_k-1\right)^2.
\end{aligned}
\end{array}\right.
\end{equation}
When $n=2$, the above system of equations reduce to the one used in \cite{dai2022asymptotics, xu2020gap}.

As the first result, we obtain an integral representation of the determinant  \eqref{deformedfred1} by using the Hamiltonian associated with a family of solutions of the coupled Painlev\'e V system \eqref{eq:CPV}.

\newtheorem{thm}{THEOREM}
\begin{thm}\label{integralexpression}
    Let  $\alpha>-\frac{1}{2}$, $\beta \in i\mathbb{R}$, $n \geq 1$, $\vec{\gamma}=\left(\gamma_0, \gamma_1, \dots, \gamma_{n-1}\right) \in[0,1)^n$  and $\vec{r}=\left(r_0, r_1, \dots, r_n\right) \in \mathbb{R}^{n+1}$ be such that $r_0< \dots <r_m=0< \dots <r_n$,  we  have the following integral expression of the Fredholm determinant defined in \eqref{deformedfred1}
    \begin{equation}\label{eq:ie}
\ln F(t\Vec{r},\Vec{\gamma};\alpha,\beta)= \int_0^t H(s)ds,
\end{equation}
where the Hamiltonian $H$ is defined by (\ref{hamiltonform}) and (\ref{defhamilton})  subject to the following asymptotics:
as $t\rightarrow0^+$
\begin{equation}\label{eq:Hasyzero}
     H(t)=i \sum_{k=0}^n \frac{c_k \cdot \Gamma(1+\alpha-\beta) \Gamma(1+\alpha+\beta) \left(2| r_j|\right)^{2 \alpha+1}}{(2\alpha+1)\Gamma^2(1+2 \alpha)} t^{2 \alpha}+\mathcal{O}\left(t^{2 \alpha+1}\right),
\end{equation}
with
\begin{equation}\label{defck}
    c_k= \begin{cases}\frac{\gamma_{k-1}-\gamma_k}{2 \pi i} e^{\beta \pi i}, & 0\leq k<m, \\ \frac{\gamma_k-\gamma_{k-1}}{2 \pi i} e^{-\beta \pi i}, & m<k\leq n,\end{cases}
\end{equation}
and as $t\rightarrow+\infty$
\begin{equation}\label{hamiltonsym0}
     H(t)=\sum_{k=0}^n 2 i b_k r_k-\left({\sum_{k=0}^n b_k^2+2 \beta b_m}\right ){t^{-1}}+\mathcal{O}\left(t^{-2}\right),
\end{equation}
with
\begin{equation}\label{def bk}
    b_k=\frac{1}{2 \pi i} \ln \frac{1-\gamma_{k-1}}{1-\gamma_k},
\end{equation}
where $\gamma_{-1}=\gamma_n=0$.
\end{thm}

\newtheorem{rmk}{REMARK}
\begin{rmk}\label{rmk1}
For $n=2$,  with the special parameter {$r_0=-1, r_1=0, r_2=1$ and $\gamma_0=\gamma_1=\gamma$, (\ref{sigma}) is reduced into $\sigma(x)=\gamma_0\chi_{(-t,0)}(x)+\gamma_1\chi_{(0,t)}(x)=\gamma \chi_{(-t,t)}(x)$}. The equations \eqref{eq:ie} and \eqref{hamiltonsym0}
 recover the integral representation
of the determinant $F((-t,0,t), (\gamma,\gamma);\alpha,\beta)$, which was derived by Dai and Zhai in \cite{dai2022asymptotics}. While  the determinant $F((-t,0,t), (1,1);\alpha,\beta)$ was expressed in terms of  another solution of the coupled Pailev\'e V system characterized  by different asymptotic behavior as $t\to +\infty $, as shown by the first and third authors of the paper in \cite{xu2020gap}.
For general $n$, the integral representation \eqref{eq:ie} is also true when some of the parameters $\gamma_j=1${, but} the asymptotic behavior would be quite different from \eqref{hamiltonsym0}. We will consider this problem in a separate paper.
\end{rmk}

\begin{thm}\label{thm large asymp}
Under the same assumptions as in Theorem \ref{integralexpression}, we have the  asymptotic approximation as $t\rightarrow+\infty$
\begin{equation}\label{large gap asym}
\begin{aligned}
\ln F(t \vec{r}, \vec{\gamma};\alpha,\beta) = & \sum_{k \neq m} 2 i b_k r_k t +\sum_{k \neq m}\left(2 \beta b_k-2 b_k^2\right) \ln \left|2 r_k t\right| - \sum_{0 \leq j<k \leq n ;\; j, k \neq m}2 b_j b_k \ln \left|\frac{2 r_j r_k t}{r_k-r_j}\right|\\
& -\frac{\alpha}{2} \ln \left[\left(1-\gamma_{m-1}\right)\left(1-\gamma_m\right)\right]+\ln \frac{G\left(1+\alpha+\beta+b_m\right) G\left(1+\alpha-\beta-b_m\right)}{G(1+\alpha+\beta) G(1+\alpha-\beta)} \\
& +\sum_{k \neq m} \ln \left[G\left(1+b_k\right) G\left(1-b_k\right)\right]+\mathcal{O}\left(\frac{1}{t}\right),
\end{aligned}
\end{equation}
where {$b_k$ is defined by (\ref{def bk})} and $G$ is the Barnes  $G$-function.{ The asymptotics (\ref{large gap asym}) are uniform for $\gamma_k$ in compact subsets of $[0,1)$ and
uniform for $r_k$ in compact subsets of $\mathbb{R}$, as long as there exists
$\epsilon >0$ such that $\min_{0\leq j <k\leq n} (r_k-r_j)>\epsilon$.}

Furthermore, (\ref{large gap asym}) can be differentiated any $k$ times with respect to $\gamma_j$ in the form $\prod_{j=0}^{n-1}\partial_{\gamma_j}^{k_j}$, where $k_0, \dots, k_{n-1} \in \mathbb{N}\cup \{0\}$ and $k=\sum_{j=0}^{n-1}k_j$, with the error term replaced by $\mathcal{O}\left(\frac{(\ln t)^k}{t}\right)$.
\end{thm}

For $n=2$ with the special parameter $r_0=-1, r_1=0, r_2=1$ and $\gamma_0=\gamma_1=\gamma\in[0,1)$, we recover the asymptotics of the determinant
$F((-t,0,t), (\gamma,\gamma);\alpha,\beta)$ in \cite{dai2022asymptotics}.
\newtheorem{col}{COROLLARY}

\begin{col}\label{cor on symmetric interval} {\rm(Dai and Zhai \cite{dai2022asymptotics})}
    For $\alpha>-\frac{1}{2}$, $\beta \in i\mathbb{R}$ and  $ \gamma\in[0,1)$,    let $\left.K^{(\alpha,\beta)}\right|_{(-t,t)}$ be the integral operator acting on $L^2(-t,t)$ with the confluent hypergeometric  kernel defined in \eqref{chfkernel}. Then, as $t\rightarrow +\infty$, we have
\begin{equation}\label{large gap symmetric interval}
\ln  \det \left(1-\gamma \left.K^{(\alpha,\beta)}\right|_{(-t,t)}\right)=-4ct+2c^2 \ln (4t) +2\alpha \pi c + 2\ln \left[G(1+ic)G(1-ic)\right]+\mathcal{O}\left(\frac{1}{t}\right),
\end{equation}
{where $c=-\frac{1}{2\pi}\ln(1-\gamma)$.}
\end{col}

While for the special parameter $\alpha=\beta=0$, we recover the asymptotics of the determinant of the sine kernel with several discontinuities in \cite{charlier2021large}.

\begin{col}\label{cor for sine}{\rm(Charlier \cite{charlier2021large})}
 Let  $K_{\sigma}^{(0,0)}$ denotes the integral operator acting on $L^2(tr_0, tr_n)$ with the sine kernel multiplied by the function $\sigma$ given in \eqref{sigma}.
   Then, as $t\rightarrow +\infty$, we have
\begin{equation}\label{large gap sine}
      \begin{aligned}
\ln \det(I-K_{\sigma}^{(0,0)})= & \sum_{k \neq m} 2 i b_k r_k t -\sum_{k \neq m}2 b_k^2 \ln \left|2 r_k t\right|  -\sum_{0 \leq j<k \leq n ;\; j, k \neq m}2 b_j b_k \ln \left|\frac{2 r_j r_k t}{r_k-r_j}\right|\\
& +\sum_{k = 0}^n \ln \left[G\left(1+b_k\right) G\left(1-b_k\right)\right]+\mathcal{O}\left(\frac{1}{t}\right).
\end{aligned}
    \end{equation}
\end{col}

\subsection{Applications}
Consider the counting function $N(t)$ of the number of particles in the confluent hypergeometric process falling into the interval $[0,t]$ for $t>0$ and  $[t,0]$ for $t<0$.  As shown in \eqref{exponential moments}, the moment generating function for the counting functions can be expressed in terms of the    confluent hypergeometric-kernel determinant. By using Theorem \ref{thm large asymp},  we derive the asymptotics for the mean, variance and covariance of the counting functions.
\begin{col}\label{meanvarcov1}
    For  $r_2>r_1>0$ fixed, as $t\rightarrow+\infty$, we have
     \begin{equation}\label{mean1}
 \mathbb{E}\left(N(t r_1)\right)={\mu_\alpha(t r_1)+\frac{\beta}{i\pi}\delta(t r_1)}+\theta_{\alpha, \beta, 1} +\mathcal{O}\left(\frac{\ln t}{t}\right)  ,
    \end{equation}
    \begin{equation}\label{mean2}
 \mathbb{E}\left(N(-t r_1)\right)=\mu_\alpha(t r_1)-\frac{\beta}{i\pi}\delta(t r_1)-\theta_{\alpha, \beta, 1} +\mathcal{O}\left(\frac{\ln t}{t}\right) ,
    \end{equation}
   \begin{equation}\label{var1}
        \operatorname{Var}\left(N(\pm t r_1)\right)=\frac{\delta(t r_1)-{(\ln G)}^{''}(1)}{ \pi^2}+\theta_{\alpha, \beta, 2}+\mathcal{O}\left(\frac{{(\ln t)^2}}{t}\right),
    \end{equation}
    \begin{equation}\label{cov1}
        \operatorname{Cov}\left(N(t r_1), N(t r_2)\right)=\Sigma(t r_1,t r_2)+\theta_{\alpha, \beta, 2}+\mathcal{O}\left(\frac{(\ln t)^2}{t}\right),
    \end{equation}
\begin{equation}\label{cov2}
        \operatorname{Cov}\left(N(t r_1), N(-t r_2)\right)=-\Sigma(t r_1,t r_2)-\theta_{\alpha, \beta, 2}+\mathcal{O}\left(\frac{(\ln t)^2}{t}\right),
    \end{equation}
where
        \begin{equation}\label{defmudeltasigma}
        \mu_\alpha(x)=\frac{x}{\pi}-\frac{\alpha}{2},\quad \delta(x)= \ln 2x,\quad\Sigma(x,y)=\frac{1}{2\pi^2}\ln\frac{2xy}{|x-y|}, {\quad \text{for }x,y>0.}
    \end{equation}
Here the parameters \begin{equation}\label{deftheta1}
    \theta_{\alpha, \beta, 1}=\frac{1}{2 \pi i} \left({(\ln G)}^{'}(1+\alpha-\beta)-(\ln G)^{'}(1+\alpha + \beta )\right),
\end{equation}
\begin{equation}\label{deftheta2}
    \theta_{\alpha, \beta, 2}=-\frac{1}{4 \pi^2} \left((\ln G)^{''}(1+\alpha +\beta)+(\ln G)^{''}(1+\alpha -\beta)\right),
\end{equation}
and  $-{(\ln G)}^{''}(1)-1=\gamma_E \approx 0.5772$ is  Euler's  constant.
\end{col}

\begin{proof}
From (\ref{exponential moments}) with  $n=1$ and the parameters $0=r_0<r_1$,  $b_m=b_0=-b$,
we obtain the following equation
\begin{equation}\label{expesionofexpo}
    F(t(0,r_1),1-e^{2\pi i b};\alpha,\beta)=\mathbb{E}\left(e^{2\pi i b N(tr_1)}\right).
\end{equation}
{Thus, we derive that}
\begin{equation}\label{Fr1mean}
    {\mathbb{E}\left(N(tr_1)\right)=\frac{1}{2\pi i}\left.\partial_b\ln F(t(0,r_1),1-e^{2\pi i b};\alpha,\beta)\right|_{b=0},}
\end{equation}
{and}
\begin{equation}\label{Fr1var}
    {\operatorname{Var}\left(N(tr_1)\right)=-\frac{1}{4\pi^2}\left.\partial_b^2\ln F(t(0,r_1),1-e^{2\pi i b};\alpha,\beta)\right|_{b=0}.}
\end{equation}
Similarly, we compute the covariance between $N(tr_1)$ and $N(tr_2)$ as follows
\begin{equation}\label{Fr1r2cov}
{\operatorname{Cov}\left(N(tr_1),N(tr_2)\right)=-\frac{1}{8\pi^2}\left.\partial_b^2\ln \left(\frac{F\left(t(0, r_1, r_2),(1-e^{4 \pi i b}, 1-e^{2 \pi i b}) ; \alpha, \beta\right)}{F\left(t(0, r_1), 1-e^{2 \pi i b} ; \alpha, \beta\right) F\left(t(0, r_2), 1-e^{2 \pi i b} ;\alpha, \beta\right)}\right)\right|_{b=0}.}
\end{equation}

By using the definition of the Barnes  $G$-function (\cite[(5.17.3)]{NIST:DLMF}), we have
\begin{equation}\label{defGfunction}
    G(1+z)=(2 \pi)^{z / 2} \exp \left(-\frac{1}{2} z(z+1)-\frac{1}{2} \gamma_E z^2\right) \prod_{k=1}^{\infty}\left(\left(1+\frac{z}{k}\right)^k \exp \left(-z+\frac{z^2}{2 k}\right)\right).
\end{equation}
 From Theorem \ref{thm large asymp} and \eqref{Fr1mean}-\eqref{defGfunction}, we obtain after some direct calculations the mean, variance of $N(r_1t)$ and the covariance between $N(tr_1)$ and $N(tr_2)$ as given in (\ref{mean1}), (\ref{var1}) and (\ref{cov1}).

The mean and variance of $N(-r_1t)$ in (\ref{mean2}) and  (\ref{var1}) can be derived in the same manner by considering
 (\ref{exponential moments}) with  $n=1$, specifying  the parameters $t(-r_1,0)$ and $b_m=b_1=-b$. The formula \eqref{cov2} can be derived similarly by considering $F((-tr_2,0,tr_1),(1-e^{2\pi i b},1-e^{2\pi i b});\alpha,\beta)$ like \eqref{Fr1r2cov}. We complete the proof of the corollary.
\end{proof}

From the definition of the counting function $N(t)$, the  number of particles in the confluent hypergeometric process falling into the interval $[-t,t]$ is given by $N_0(t):=N(t)+N(-t)$.  Combining  (\ref{mean1}), (\ref{var1}), (\ref{mean2}) and (\ref{cov2}), we immediately derive the asymptotics of the mean and variance  of the counting function $N_0(t)$, { which recovers the one established by Dai and Zhai in \cite{dai2022asymptotics}.}
\begin{col}\label{meanvarforsyminterval}{\rm(Dai and Zhai \cite{dai2022asymptotics})}
We have as $t\rightarrow+\infty$,
    \begin{equation}
    \mathbb{E}(N_0(t))=\frac{2t}{\pi}-\alpha +\mathcal{O}\left(\frac{\ln t}{t}\right),\quad \operatorname{Var}(N_0(t))=\frac{\ln 4t +1 + \gamma_E}{\pi^2}+\mathcal{O}\left(\frac{(\ln t)^2}{t}\right),
    \end{equation}
    where $\gamma_E \approx 0.5772$ is  Euler's  constant.
\end{col}

The rest of the paper is organized as follows.  In Section \ref{CPV}, we formulate an RH problem  for the  Fredholm determinant defined in \eqref{deformedfred1} by using the general theory of integrable operator developed in \cite{its1990differential}.  Making  use of the confluent hypergeometric functions, we transform the RH problem to a model RH problem $Y(\eta,t )$ with constant jumps. We then derive a Lax pair for $Y(\eta,t)$ and the compatibility conditions of the Lax pair lead us to  the coupled Painlev\'e V system. The Hamiltonian formulation of the system and several useful differential identities are also established in this section.  In Sections \ref{sec:Asy0} and \ref{sec:Asyinfty},
 we  derive the asymptotic approximations of a family of solutions to the coupled Painlev\'e V system and the associated Hamiltonian as $t\rightarrow 0^+$ and $t\to +\infty$ by performing Deift-Zhou nonlinear steepest descent analysis  \cite{deiftzhou1993, deift2} of the RH problem for $Y(\eta,t)$. The results are stated in Proposition \ref{pro: CPVasy0} and \ref{pro: CPVasyinfty} at the end of Sections \ref{sec:Asy0} and \ref{sec:Asyinfty}, respectively.  The last section, Section \ref{sec:proofs}, is devoted to the proofs of Theorems \ref{integralexpression} and \ref{thm large asymp}. For the convenience of the reader,
a brief description of the confluent hypergeometric parametrix is put in the Appendix.

\section{Model RH problem and the coupled Painlev\'e V system}\label{CPV}
In this section, we first relate the  determinant \eqref{deformedfred1}  to  the solution of an RH problem by using the general theory of integrable operator developed in \cite{its1990differential}. Then, we derive from the RH problem  the Lax pair and the Hamiltonian for the coupled Painlev\'e V system \eqref{eq:CPV}. These results in the present  section will be used in Section \ref{sec:proofs} to prove Theorem \ref{integralexpression}.

\subsection{RH problem for the Fredholm determinant}

The kernel of the operator $K_{t\Vec{r},\Vec{\gamma},\alpha,\beta}$ can be expressed in the following   from
\begin{equation}\label{integrable form of kernel}
    \sigma(y) K^{(\alpha, \beta)}(x,y)=\frac{\vec{f}(x)^{T} \vec{g}(y)}{x-y},
\end{equation}
where
\begin{equation}\label{f, g with chf}
    \vec{f}(z)=\frac{1}{\sqrt{2 \pi i}}\left[\begin{array}{c}
\frac{\Gamma(1+\alpha-\beta)}{\Gamma(1+2 \alpha)} \mathbb{B}(z) \\
\frac{\Gamma(1+\alpha+\beta)}{\Gamma(1+2 \alpha)} \mathbb{A}(z)
\end{array}\right]
,\quad
\vec{g}(z)=\frac{\sigma(z)}{\sqrt{2 \pi i}}\left[\begin{array}{c}
-\frac{\Gamma(1+\alpha+\beta)}{\Gamma(1+2 \alpha)} \mathbb{A}(z) \\
\frac{\Gamma(1+\alpha-\beta)}{\Gamma(1+2 \alpha)} \mathbb{B}(z)
\end{array}\right],
\end{equation}
with $\mathbb{A}(z)$ and  $\mathbb{B}(z)$ defined in \eqref{eq:kernelA} and \eqref{eq:chf}.
Therefore,  the kernel  is integrable in the sense of  \cite{its1990differential}.
By standard properties of trace class operators, we have
\begin{equation}\label{fred with resolvent}
    \begin{aligned}
\frac{d}{d t} \ln F & =\frac{d}{d t} \operatorname{tr}\left(\ln \left(1-K_{\sigma}^{(\alpha,\beta)}\right)\right) \\
& =-\operatorname{tr}\left(\left(1-K_{\sigma}^{(\alpha,\beta)}\right)^{-1} \frac{d}{d t} K_{\sigma}^{(\alpha,\beta)}\right) \\
& =\sum_{k=0}^n r_k\left(\lim _{z \rightarrow (r_kt)^{+} } R(z, z)-\lim _{z \rightarrow (r_kt)^{-} } R(z, z)\right),
\end{aligned}
\end{equation}
where $R(x,y)$ denotes the kernel of the resolvent operator  $(1-K_{\sigma}^{(\alpha,\beta)})^{-1}K_{\sigma}^{(\alpha,\beta)}$. According to \cite{its1990differential} and \cite{DIZ}, the kernel is also of the integrable form
\begin{equation}\label{defres}
    R(x,y)= \frac{\left(\Vec{F}(x),\Vec{G}(y)\right)}{x-y},
\end{equation}
with
\begin{equation}\label{def F, G}
    \vec{F}(z)=\left(1-K_{\sigma}^{(\alpha,\beta)}\right)^{-1} \vec{f}(z),\quad
    \vec{G}(z)=\left(1-K_{\sigma}^{(\alpha,\beta)}\right)^{-1} \vec{g}(z).
\end{equation}
Furthermore,  $ \vec{F} $  and $\vec{G}$  can be expressed as
\begin{equation}\label{F with f}
    \vec{F}(z)=m_{+}(z) \vec{f}(z), \quad
     \vec{G}(z)=m_{+}(z)^{-T} \vec{g}(z).\end{equation}
Here $m(z)$ solves   the following RH problem.

\subsection*{RH problem for $m(z)$}
\begin{itemize}
    \item[\rm (1)]  $m(z)$ analytic in $\mathbb{C} \backslash [r_0t,r_nt]$,
    \item[\rm (2)] $m_{+}(z)=m_{-}(z)\left(I-2 \pi i \vec{f}(z) {\vec{g}(z)}^T\right )$,\quad on $\displaystyle{\cup_{k=0}^{n-1}\left(r_k t, r_{k+1} t\right)}$,
    \item[\rm (3)] $m(z)\rightarrow I$ ,\quad as $z\rightarrow\infty$.
\end{itemize}
Conversely, we have \cite{its1990differential}
\begin{equation}\label{solution m}
    m(z)=I-\int_{r_0t}^{r_nt} \frac{\vec{F}(s)  {\vec{g}(s)}^T}{s-z}ds,
\end{equation}
with $\vec{F}(s) $ defined by using the inverse of the operator $I-K_{\sigma}^{(\alpha,\beta)}$ as given in \eqref{def F, G}. The determinant $\det(1-K_{\sigma}^{(\alpha,\beta)})$  is  strictly positive since it represents  the gap probability of the thinned confluent hypergeometric process. Therefore, the operator $I-K_{\sigma}^{(\alpha,\beta)}$ is invertible and hence the existence of solution to the RH for $Y(z)$   is justified.

\subsection{Model RH problem}

Next, we transform the RH problem for $m(z)$ to a new RH problem with constant jumps.  To begin with,  we note that the function $\vec{f}(z)$ and $\vec{g}(z)$ defined in \eqref{f, g with chf}
can be expressed as

\begin{equation}\label{theta with f}
 \vec{f}(z)=   e^{\frac{\beta}{2} \pi i \sigma_3} \Phi(2 z) e^{\pm \frac{\alpha}{2} \pi i \sigma_3} e^{-\frac{\beta}{2} \pi i \sigma_3}\left(\frac{1}{\sqrt{2 \pi i}}\right)^{\sigma_3}\left[\begin{array}{l}
1 \\
0
\end{array}\right],\end{equation}
and
\begin{equation}\label{theta with g}
 \vec{g}(z)=  \sigma(z) e^{-\frac{\beta}{2} \pi i \sigma_3}\left( \Phi(2 z)\right)^{-T} e^{\mp \frac{\alpha}{2} \pi i \sigma_3} e^{\frac{\beta}{2} \pi i \sigma_3}\left(\frac{1}{\sqrt{2 \pi i}}\right)^{-\sigma_3}\left[\begin{array}{l}
0 \\
1
\end{array}\right],\end{equation}
for $z\in \Omega_{\Phi,1} $ and $\pm \mathtt{Re} z>0$. Here $ \Phi(z)$ is the solution of the RH problem for the confluent hypergeometric parametrix given in the Appendix.

We introduce the transformation
\begin{equation}\label{def model RH}
  Y(z, t)=(2 t)^{\beta \sigma_3} e^{-\frac{\beta}{2} \pi i \sigma_3} m(t z) e^{\frac{\beta}{2} \pi i \sigma_3} \left\{\begin{aligned}
&\Phi(2tz), \quad z\in\Omega_{Y,j} ,~~ j=2,5,\\
& \widehat{ \Phi}_j(2tz),  \quad z\in\Omega_{Y,j}, ~~ j=1,3,4,
\end{aligned}
\right.
\end{equation}
where $ \Phi(z)$ is the solution of the RH problem for the confluent hypergeometric parametrix given in the Appendix. The functions  $\widehat{ \Phi}_j(z)$ denote the analytic extension of $ \Phi(z)$ from the regions $\Omega_{\Phi,j}$ to $\Omega_{Y,j}$, for $j=1,3,4$, respectively; see Figures \ref{fig:3} and  \ref{fig:CHF} for the regions.

Then $Y(z)=Y(z,t)$ satisfies the following RH problem.
\subsection*{The model RH problem for $Y(z)$}
\begin{itemize}
    \item[\rm (1)] $Y(z)$ is analytic for $z\in\mathbb{C}\backslash\cup{_{j=1}^7\Gamma_j}$, where the oriented contour are given by
    $$\begin{aligned}
& \Gamma_1=r_n+e^{\frac{\pi i}{4}} \mathbb{R}^{+},~~ \Gamma_2=r_0+e^{\frac{3 \pi i}{4}} \mathbb{R}^{+}, ~~\Gamma_3=r_0+e^{-\frac{3 \pi i}{4}} \mathbb{R}^{+},~~ \Gamma_4=e^{-\frac{\pi i}{2}} \mathbb{R}^{+}, \\
& \Gamma_5=r_n+e^{-\frac{\pi i}{4}} \mathbb{R}^{+},~~ \Gamma_6=\cup_{k=0}^{m-1}\left(r_k, r_{k+1}\right), ~~\Gamma_7=\cup_{k=m}^{n-1}\left(r_k, r_{k+1}\right),
\end{aligned}$$ as shown in Figure \ref{fig:3}.
\begin{figure}[h]
    \centering
    \includegraphics[width=0.7\textwidth]{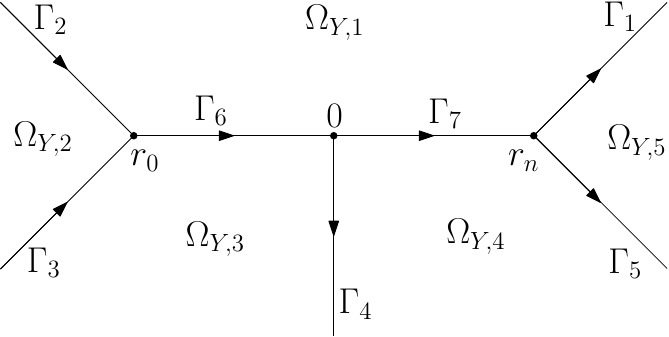}
    \caption{\small{The jump contours and regions of the model RH problem for $Y$}}
    \label{fig:3}
\end{figure}

\item[\rm (2)]  $Y_+(z)=Y_-(z)J_Y(z)$  on $\cup{_{j=1}^7\Gamma_j}$ , where
\begin{equation}\label{jump4Y}
    J_Y(z)=\left\{\begin{array}{ll}
{\left[\begin{array}{cc}
1 & 0 \\
e^{-(\alpha-\beta) \pi i} & 1
\end{array}\right]},& {z \in \Gamma_1}, \\
{\left[\begin{array}{cc}
1 & 0 \\
e^{(\alpha-\beta) \pi i} & 1
\end{array}\right]},& {z \in \Gamma_2} ,\\
{\left[\begin{array}{cc}
1 & -e^{-(\alpha-\beta) \pi i} \\
0 & 1
\end{array}\right]}, & { z \in \Gamma_3}, \\
{e^{2 \beta \pi i \sigma_3}},& {z \in \Gamma_4 },\\
{\left[\begin{array}{cc}
1 & -e^{(\alpha-\beta) \pi i} \\
0 & 1
\end{array}\right]},& { z \in \Gamma_5}, \\
{\left[\begin{array}{cc}
0 & -e^{-(\alpha-\beta) \pi i} \\
e^{(\alpha-\beta) \pi i} & 1-\sigma(z t)
\end{array}\right]},& {z \in \Gamma_6}, \\
{\left[\begin{array}{cc}
0 & -e^{(\alpha-\beta) \pi i} \\
e^{-(\alpha-\beta) \pi i} & 1-\sigma(z t)
\end{array}\right]},& {z \in \Gamma_7}.
\end{array}\right.
\end{equation}

\item[\rm (3)]   As $z\rightarrow \infty$, we have for  $\arg z\in(-\frac{\pi}{2},\frac{3\pi}{2})$

\begin{equation}\label{Y asym at inf}
    Y(z)=\left(I+\frac{Y_1(t)}{z}+\frac{Y_2(t)}{z^2}+\mathcal{O}\left(z^{-3}\right)\right) z^{-\beta \sigma_3} e^{-i z t \sigma_3}.
\end{equation}

\item [\rm (4)]  As $z\rightarrow r_k$ , $k\neq m$ ,  we have
\begin{equation}\label{Y asym at rk}
    Y(z)=Y_{0}^{(k)}(t)\left(I+Y_1^{(k)}(t)\left(z-r_k\right)+Y_2^{(k)}\left(t\right)\left(z-r_k\right)^2+\ldots\right)\left[\begin{array}{cc}
1 & \Tilde{c}_k \ln \left(z-r_k\right) \\
0 & 1
\end{array}\right],
\end{equation}
where $\Tilde{c}_k=\operatorname{sgn}(r_k) c_ke^{\operatorname{sgn}(r_k)\alpha \pi i}$ and $c_k$ is defined in (\ref{defck}) and $z\in \Omega_{Y,1}$. The branch for $\ln \left(z-r_k\right)$ is taken to be $\arg(z-r_k)\in (0, 2\pi )$  if $r_k<0$, or $\arg(z-r_k) \in (-\pi , \pi )$  if $r_k > 0$. The behavior of $Y(z)$ as $z\rightarrow r_k$  for $z$ in other regions is determined by
using  the jump condition \eqref{jump4Y}.

\item[\rm (5)]   As $z \rightarrow 0$ , we have
\begin{equation}\label{Y asym at 0}
    Y(z)=Y_0^{(m)}(t)\left(I+Y_1^{(m)}(t) z+Y_2^{(m)}(t) z^2+\ldots\right) z^{\alpha \sigma_3}\left\{\begin{array}{l}
{C_{m,1}} ,  2\alpha \notin \mathbb{Z}^+ \cup \{0\},\\
{C_{m,2},  2\alpha \in \mathbb{Z}^+ \cup \{0\}},
\end{array}\right.
\end{equation}
where $z\in \Omega_{Y,1}$,
$$C_{m,1} = \left[\begin{array}{cc}
1 & \frac{c_m}{2 i \sin 2 \alpha \pi} \\
0 & 1
\end{array}\right],\quad C_{m,2}=\left[\begin{array}{cc}
1 & \frac{e^{-2 \alpha \pi i}}{2 \pi i}c_m \ln z \\
0 & 1
\end{array}\right],
$$
and
$$c_m=\left(1-\gamma_{m-1}\right) e^{(\alpha+\beta) \pi i}-\left(1-\gamma_m\right) e^{-(\alpha+\beta) \pi i}.$$
Both $z^\alpha$ and $\ln z$ take the principal branches. The behavior of $Y(z)$ as $z\rightarrow 0$  for $z$ in other regions is determined by using  the jump condition \eqref{jump4Y}.
\end{itemize}

\subsection{Lax pair and the coupled Painlev\'e V system}
In this section, we derive a Lax pair from the above model RH problem, of which the compatibility condition gives us the coupled Painlev\'e V system. The Hamiltonian for the  coupled Painlev\'e V system is also derived.

\newtheorem{pro}{PROPOSITION}
\begin{pro}\label{lax pair and cpv}
   We have the Lax pair
       \begin{equation}\label{lax pair}
    \frac{\partial}{\partial z} Y(z, t)= \mathcal{A}(z, t)Y(z, t), \quad \frac{\partial}{\partial t} Y(z, t)= \mathcal{B}(z, t)Y(z, t),
\end{equation}
where \begin{equation}\label{defA}
     \mathcal{A}(z, t)=\sum_{k=0}^n \frac{A_k(t)}{z-r_k}-i t
     \sigma_3,
\end{equation}
\begin{equation}\label{defB}
    \mathcal{B}(z, t)=-i z \sigma_3+B(t),
   \end{equation}
with the coefficients being  given as
\begin{equation}\label{defAm}
    A_m(t)=\left[\begin{array}{cc}
-\beta-\sum_{k \neq m} u_k(t)v_k(t) & \left(\alpha+\beta+\sum_{k \neq m} u_k(t)v_k  (t)\right) y(t) \\
\frac{\alpha-\beta-\sum_{k \neq m}u_k(t)v_k  (t)}{y(t)} & \beta+\sum_{k \neq m} u_k(t)v_k  (t)\end{array}\right],
\end{equation}
\begin{equation}\label{defAk}
    A_k(t)=\left[\begin{array}{cc}
-u_k(t) v_k(t) & u_k(t) y(t) \\
-\frac{u_k(t) v_k^2(t)}{y(t)} & u_k(t) v_k(t)
\end{array}\right], \quad k \neq m,
\end{equation}
and
\begin{equation}\label{defD}
    B(t)=\frac{1}{t}\left[\begin{array}{cc}
0 & d_1(t) y(t) \\
\frac{d_2(t)}{y(t)} & 0
\end{array}\right],
\end{equation}
where
\begin{equation}\label{defd1d2}
    \begin{aligned}
        d_1(t)=\alpha+\beta-\sum_{k \neq m} u_k(t)\left(v_k(t)-1\right),\\
        d_2(t)=\alpha-\beta-\sum_{k \neq m} u_k(t) v_k(t)\left(v_k(t)-1\right).
    \end{aligned}
\end{equation}
The compatibility condition of the Lax pair is expressed as the coupled Painlev\'e V system (\ref{eq:CPV})  and the following differential equation for
$y(t)$
\begin{equation}\label{diff4dandy}
     t \frac{d}{d t} y=y\left(d_1-d_2\right).
\end{equation}
Defining  the  Hamiltonian by
\begin{equation}\label{defH}
    tH=2 i t {(Y_1)}_{11}-(\alpha^2-\beta^2),
\end{equation}
where ${(Y_1)}_{11}$ is the $(1,1)$-entry of the coefficient matrix $Y_1$ in the asymptotic approximation (\ref{Y asym at inf}) as $z\rightarrow\infty$, then we have
\begin{equation}\label{Hexp}
    t H=\sum_{k \neq m}-2 i t r_k H_V\left(u_k, v_k,-2 i t r_k ; \alpha, \beta\right)+\frac{1}{2} \sum_{j \neq k;\; j, k \neq m} u_j u_k\left(v_j+v_k\right)\left(v_j-1\right)\left(v_k-1\right),
\end{equation}
where $H_V$ is
 the classical Hamiltonian for the Painlev\'e V equation defined in \eqref{eq:hamilton}. The Hamiltonian system \eqref{hamiltonform} is equivalent to  the   coupled Painlev\'e V system \eqref{eq:CPV}.
\end{pro}

\begin{proof} Observing that the jumping matrices in (\ref{jump4Y}) are independent of both the variables $z$ and $t$ , $Y(z,t)$ , $\frac{d}{dz}Y(z,t)$ and $\frac{d}{dt}Y(z,t)$ have the same jump conditions. Thus, $\mathcal{A}$ and $\mathcal{B}$
are meromorphic functions in the $z$-plane with the only possible isolated singularities at $z=r_k$, $k=0,\dots,n$. Taking  \eqref{Y asym at inf},\eqref{Y asym at rk} and  \eqref{Y asym at 0} into account, we then obtain that  $\mathcal{A}$ and $\mathcal{B}$ are of  the form  given  in \eqref{defA} and \eqref{defB}.

Next, we derive the coefficients in the $\mathcal{A}$ and $\mathcal{B}$.  It is seen from  \eqref{Y asym at inf},  \eqref{Y asym at rk}) and \eqref{Y asym at 0} that
\begin{equation}\label{det Ak}
   \det A_m=-\alpha^2, \qquad   \det A_k=0, ~ k\neq m.
\end{equation}
Moreover,  in view of \eqref {Y asym at inf}, we obtain by comparing the coefficients of $1/z$ on both sides of the first equation in \eqref{lax pair}
\begin{equation}\label{sumAk}
 \sum_{k=0}^nA_k=\left[\begin{array}{cc}
-\beta & 2it(Y_1)_{12} \\
-2it (Y_1)_{21}& \beta
\end{array}\right].
\end{equation}
Therefore, $A_k$, $k=0, \dots n$ can be parametrized as \eqref{defAm} and \eqref{defAk}.
Similarly, we obtain \eqref{defD} by substituting \eqref {Y asym at inf} into the second equation of \eqref{lax pair} and using \eqref{sumAk}.

The compatibility condition of the Lax pair gives
\begin{equation}\label{compability condition}
    \begin{aligned}
 \frac{d}{d t} A_k(t)  =\left[-i r_k \sigma_3+B(t), A_k(t)\right], \quad k\neq m,
\end{aligned}
\end{equation}
and
\begin{equation}\label{compability condition2}
 \frac{d}{d t} A_m(t)  =\left[B(t), A_m(t)\right].
\end{equation}
We then have from \eqref{compability condition}
\begin{equation}\label{diff eq cpv}
    \left\{\begin{array}{l}
t \frac{d}{d t}\left(u_k v_k\right)=d_1 u_k v_k^2+d_2 u_k , \quad k\neq m,\\[.3cm]
t \frac{d}{d t}\left(u_k y\right)=-2 i t r_k u_k y+2 d_1 u_k v_k y,  \quad k\neq m.
\end{array}\right.
\end{equation}
While  the equation \eqref{compability condition2} gives us
\begin{equation}\label{diff eq am,y,d}
    \left\{\begin{array}{l}
t \frac{d}{d t} {(A_m)}_{11}={(A_m)}_{11}\left(d_1+d_2\right)+\alpha\left(d_1-d_2\right), \\[.3cm]
t \frac{d}{d t} y=y\left(d_1-d_2\right),\\
\end{array}\right.
\end{equation}
with $d_1$ and $d_2$ given in \eqref{defd1d2}.
Combining \eqref{diff eq cpv} and \eqref{diff eq am,y,d}, we arrive at the coupled Painlev\'e V system \eqref{eq:CPV}  and \eqref{diff4dandy}.

To derive the expression of the Hamiltonian, we compare the coefficient of $\frac{1}{z ^2}$ in the large $z$ expansion
of the first equation in \eqref{Y asym at 0} and obtain
\begin{equation}\label{coefficient of eta-2}
    \left[-i t \sigma_3 Y_1, Y_1\right]+\left[Y_2,-i t \sigma_3\right]-\beta\left[Y_1, \sigma_3\right]-Y_1=\sum_{k=0}^n A_k r_k.
\end{equation}
A combination of  \eqref{defH} and \eqref{coefficient of eta-2} then gives \eqref{Hexp}. From \eqref{Hexp}, it is direct to see that the Hamiltonian system \eqref{hamiltonform} is equivalent to the coupled Painlev\'e V system \eqref{eq:CPV}. This completes the proof of the proposition.
\end{proof}

For later use, we also collect several relations as what follows.
\begin{pro}\label{pro:usefulRelation}
The function $y(t)$ appears in the coefficient \eqref{defAm} is related to the solution to the RH problem for $Y(z)$ via
\begin{equation}\label{y with rh}
   y(t)=\frac{\left(Y_0^{(m)}\right)_{11}}{\left(Y_0^{(m)}\right)_{21}},
\end{equation}
where $Y_0^{(m)}$ is given in the asymptotic expansion of $Y(z)$ in \eqref{Y asym at 0}.
Moreover,  let
\begin{equation}\label{d with rh}
     d(t)=2\alpha\left(Y_0^{(m)}\right)_{11}\left(Y_0^{(m)}\right)_{21},
 \end{equation}
 then $d(t)$ satisfies the  differential equation
\begin{equation}\label{dequation} t \frac{d}{d t} d=d\left(d_1+d_2\right),
\end{equation}
with $d_1$ and $d_2$ given in \eqref{defd1d2}.
\end{pro}
\begin{proof}
From (\ref{defAm}), we have
\begin{equation}\label{eq:yA}
    y=\frac{\left(A_m\right)_{11}+\alpha}{\left(A_m\right)}.
\end{equation}
Then, by (\ref{Y asym at 0}), (\ref{det Ak}) and \eqref{eq:yA}, we obtain (\ref{y with rh}).

To derive (\ref{dequation}), we obtain  from  \eqref{Y asym at 0}  and  the  second equation in  the Lax pair (\ref{lax pair}) that
\begin{equation}\label{diffY_0}
    \frac{d}{dt}Y^{(m)}_0(t)=B(t)Y^{(m)}_0(t).
\end{equation}
This, together with (\ref{d with rh})  and (\ref{diff4dandy}), implies (\ref{dequation}).

\end{proof}

\subsection{Differential identity}\label{sec:3.4}
Recently,  the authors in \cite{bip, ilp} have developed a new method to evaluate the  total integral of the Hamiltonian. In this method, the differential identities for Hamiltonian play an important  role.  In the present  subsection, we derive the differential identities for the Hamiltonian which will be used in  Section \ref{sec:proofs}, there we applied the method developed in  \cite{bip, ilp} to calculate the constant term in the large gap asymptotics.

\begin{pro}\label{pro:4}
    We have the relation
    \begin{equation}\label{dift}
        H=\left(\sum_{k\neq m} u_k \frac{d v_k}{d t}-H\right)+\frac{d}{d t}\left(t H+\alpha \ln d-\beta \ln y-2\left(\alpha^2-\beta^2\right)\ln t\right),
    \end{equation}
    and the following differential identities
    \begin{equation}\label{difgamma}
        \frac{d}{d \gamma_j}\left(\sum_{k \neq m} u_k \frac{d v_k}{d t}-H\right)=\frac{d}{d t}\left(\sum_{k \neq m} u_k \frac{d v_k}{d \gamma_j}\right),
    \end{equation}
    with $j=0,\dots,n-1$.
\end{pro}

\begin{proof} By using the Hamiltonian dynamic system \eqref{hamiltonform}, \eqref{diff4dandy} and \eqref{dequation}, we have
\begin{equation}\label{eq for H 2}
    \begin{aligned}
\sum_{k \neq m} u_k \frac{d v_k}{d t}-H & =\sum_{k \neq m} u_k \frac{\partial H}{\partial u_k}-H \\
& =\frac{1}{t} \sum_{k \neq m} u_k^2 v_k\left(v_k-1\right)^2+\frac{1}{t} \sum_{j<k} u_j u_k\left(v_j+v_k\right)\left(v_j-1\right)\left(v_k-1\right),
\end{aligned}
\end{equation}
and
\begin{equation}\label{diff eq for t}
    \begin{aligned}
\frac{d}{d t}\left(t H+\alpha \ln d(t) \right.&\left. -\beta \ln y(t)-2\left(\alpha^2-\beta^2\right) \ln t\right) \\
& =2 i \sum_{k \neq m} u_k v_k r_k-\frac{\alpha}{t} \sum_{k \neq m} u_k\left(v_k^2-1\right)-\frac{\beta}{t} \sum_{k \neq m} u_k\left(v_k-1\right)^2.
\end{aligned}
\end{equation}
Combining \eqref{Hexp} with (\ref{eq for H 2}) and (\ref{diff eq for t}), we obtain (\ref{dift}).

Applying again the Hamiltonian dynamic system \eqref{hamiltonform}, we derive the  differential identity with respect to  the parameters \begin{equation}\label{eq:dgamma}
 \begin{aligned}
\frac{d}{d \gamma_j}\left(\sum_{k \neq m} u_k \frac{d v_k}{d t}-H\right)&
=\sum_{k \neq m}\left(\frac{d u_k}{d \gamma_j} \frac{d v_k}{d t}+u_k \frac{d^2 v_k}{d \gamma_j d t}-\frac{\partial H}{\partial u_k}\frac{d u_k}{d \gamma_j}- \frac{\partial H}{\partial v_k}\frac{d v_k}{d \gamma_j}\right)\\
&=\frac{d}{d t}\left(\sum_{k \neq m} u_k \frac{d v_k}{d \gamma_j}\right).
\end{aligned}
\end{equation}
This completes the proof of the proposition.
\end{proof}

\section{Asymptotic Analysis as $t$ $\rightarrow$ $0^+$}\label{sec:Asy0}
In this section, we  derive the asymptotic approximations of a family of solutions to the coupled Painlev\'e V system (\ref{eq:CPV}) and the associated Hamiltonian as $t\rightarrow 0^+$ by performing the Deift-Zhou nonlinear steepest descent analysis  \cite{deiftzhou1993} of the RH problem for $Y(z)$.

For the purpose, we consider
\begin{equation}\label{X with phi}
   X(z) =(2 t)^{-\beta \sigma_3} Y\left(\frac{z}{2 t}, t\right) .\end{equation}
It is readily seen from the RH problem for $Y(z)$ that $X(z)$ shares the same jump matrices with $\Phi(z)$ given in the Appendix,  for $|z|> \max\left\{2\left|r_0 t\right|, 2\left|r_n t\right|\right\}$. Therefore, the solution to the RH problem for $X(z)$ can be approximated by $\Phi(z)$ for $|z|>\max\left\{2\left|r_0 t\right|, 2\left|r_n t\right|\right\}$ bounded away from the origin as $t\to 0$.

In the next subsection, we  construct a local parametrix $L(z)$ which has the same jump condition as $X(z)$ in a neighborhood of the origin containing the interval $[2r_0 t, 2r_nt]$.

\subsection{Local parametrix}

\subsection*{RH problem for $L(z)$}
\begin{itemize}
  \item[\rm (1)] $L(z) $ is analytic for $z\in U(0, \varepsilon) \setminus \Sigma_X$ with some constant $\varepsilon>0$.
    \item[\rm (2)] We have the jump relation
      \begin{equation}\label{jumpL}L_+(z) = L_-(z)J_Y(z/2t)\end{equation} for $z\in \Sigma_X\cap U(0, \varepsilon) $, where $J_Y(z)$ is given in \eqref{jump4Y}.
    \item[\rm (2)] As $t\rightarrow 0^+$, we have
    \begin{equation}\label{LMatching}L(z) \Phi(z)^{-1}=I+O\left(t^{2\alpha+1}\right),
    \end{equation} where the error term is uniform for $z\in \partial U(0, \varepsilon)$.
\end{itemize}

Inspired by the asymptotic behavior  of $ \Phi(z)$ near the origin given in \eqref{eq: Phi0} and \eqref{eq: Phi02},  we construct $L(z)$ of the following form \begin{equation}\label{construct L}
    L(z)=\Phi^{(0)}(z) \left[\begin{array}{cc}
1 & t^{2\alpha} k\left(\frac{z}{t}\right) \\
0 & 1
\end{array}\right] z^{\alpha \sigma_3}\left\{\begin{array}{ll}
{\left[\begin{array}{cc}
1 & \frac{\sin (\alpha+\beta) \pi}{\sin 2 \alpha \pi} \\
0 & 1
\end{array}\right]C_j} ,&  2\alpha\notin \mathbb{N},\\[.4cm]
{\left[\begin{array}{cc}
1 & (-1)^{2 \alpha} \frac{\sin (\alpha+\beta) \pi}{\pi} \ln z \\
0 & 1
\end{array}\right]C_j},&  2\alpha\in \mathbb{N},
\end{array}\right.
\end{equation}
where $\alpha>-\frac{1}{2} $ and the branches for $z^{\alpha}$ and $\ln z$ are chosen such that $\arg z\in(-\frac{\pi}{2}, \frac{3}{2}\pi)$.
Here the function $k(z)$ is analytic for $z\in \mathbb{C}\setminus [2r_0, 2r_n]$ with the following jumps on $(2r_0, 2r_n)$
\begin{equation}\label{jump for k 1}
    k_{+}(x)-k_{-}(x)=-|x|^{2 \alpha} \sigma\left(\frac{xt}{2}\right) e^{(\alpha+\beta) \pi i}, \quad 2r_0<x<0,
\end{equation}
and
\begin{equation}\label{jump for k 2}
    k_{+}(x)-k_{-}(x)=-x^{2 \alpha} \sigma\left(\frac{xt}{2}\right) e^{(\alpha-\beta) \pi i}, \quad 0<x<2r_n,
\end{equation}in which $\sigma(x)$ is the step function defined in \eqref{sigma}.
Applying the Plemelj formula, we have
\begin{equation}\label{expression of k}
    k(z)=-\frac{e^{(\alpha+\beta) \pi i}}{2 \pi i} \int_{2 r_0}^0 \frac{|x|^{2 \alpha} \sigma\left(\frac{xt}{2}\right)}{x-z} d x-\frac{e^{(\alpha-\beta) \pi i}}{2 \pi i} \int_0^{2 r_n} \frac{x^{2 \alpha} \sigma\left(\frac{xt}{2}\right)}{x-z} d x, \quad  z\in \mathbb{C}\setminus [2r_0, 2r_n].
\end{equation}

It is seen from the construction \eqref{construct L} and the jump relations \eqref{jump for k 1} and \eqref{jump for k 2} that  the jump condition \eqref{jumpL} is fulfilled.

To show \eqref{LMatching}, we expand $k(z)$ as $z \rightarrow \infty$:
\begin{equation}\label{k asym at inf}
    k(z)=-e^{\alpha\pi i}\sum_{k=0}^n \frac{c_k}{2 \alpha+1}\left(2 |r_k|\right)^{2 \alpha+1} \frac{1}{z}+\mathcal{O}\left(\frac{1}{z^2}\right).
\end{equation}
Substituting (\ref{k asym at inf}) into \eqref{construct L}, we have as $t\to 0$
\begin{equation}\label{competing condition}
    L(z) \Phi{(z)}^{-1}=I-e^{\alpha\pi i}\sum_{k=0}^n \frac{c_k}{2 \alpha+1}\left(2 |r_k|\right)^{2 \alpha+1} \frac{t^{2 \alpha+1}}{z} \Phi^{(0)}(z)\left[\begin{array}{ll}
0 & 1 \\
0 & 0
\end{array}\right] {\Phi^{(0)}(z)}^{-1}+\mathcal{O}\left(t^{2 \alpha+2}\right),
\end{equation}
where the error term is uniform for $z\in \partial U(0, \varepsilon)$. Therefore, we have the matching condition \eqref{LMatching}.

\subsection{Ratio RH problem }
We define the following ratio RH problem
\begin{equation}\label{def:E}
E(z)= \begin{cases}X(z) L(z)^{-1}, & |z|<\varepsilon, \\ X(z) \Phi{(z)}^{-1}, & |z|>\varepsilon.\end{cases}
\end{equation}
Then, $E(z)$ satisfies the following RH problem.
\subsection*{RH problem for $E(z)$}
\begin{itemize}
  \item[\rm (1)] $E(z) $ is analytic for $z\in \mathbb{C} \setminus U(0, \varepsilon)$ with some constant $\varepsilon>0$.
    \item[\rm (2)]  $E(z) $ satisfies the jump condition
      \begin{equation}\label{jumpE}E_+(z) = E_-(z)J_E(z), \quad J_E(z)=\Phi{(z)}L(z)^{-1}, \end{equation} for $z\in \partial U(0, \varepsilon)$
      with the orientation clockwise.

        \item[\rm (3)] $E(z)=I+\mathcal{O}\left(1/z\right)$ as $z\to\infty$.
        \end{itemize}
Using \eqref{LMatching} and  a standard argument for the small norm RH problem as given in \cite{deiftzhou1993} , we have the expansion
\begin{equation}\label{value of E}
   E(z)=
I-e^{\alpha\pi i}\sum_{k=0}^n \frac{c_k}{2 \alpha+1}\frac{(2 |r_k|t)^{2 \alpha+1}}{z} \Phi^{(0)}(0)\left[\begin{array}{ll}
0 & 1 \\
0 & 0
\end{array}\right] {\Phi^{(0)}(0)}^{-1}+\mathcal{O}\left(t^{2 \alpha+2}\right) , \quad |z|>\varepsilon, \end{equation}
and
   \begin{equation}\label{value of E2}
   E(z)=I+\mathcal{O}\left(t^{2 \alpha+1}\right)
   \end{equation}
   for $|z|<\varepsilon$.

 \subsection{Asymptotic behavior as $t\rightarrow0^+$}
 As a consequence  of nonlinear steepest descent analysis of the RH problem for $Y(z)$, we obtain the asymptotics of  a family of special solutions to the coupled Painlev\'e V system and the Hamiltonian associated.
 \begin{pro}
   \label{pro: CPVasy0}
    For $\alpha> -\frac{1}{2}$, $\beta \in i\mathbb{R}$ and $\gamma_k \in[0,1)$, there exist  solutions to the coupled Painlev\'e V system (\ref{eq:CPV}) with the following asymptotic behaviors as $t \rightarrow 0^+$
    \begin{equation}\label{CPVasym0}
        \left\{\begin{array}{l}
u_k(t)=\frac{\operatorname{sgn}(r_k)c_k  \Gamma(1+\alpha-\beta) \Gamma(1+\alpha+\beta)(2|r_k| t)^{2 \alpha}}{\Gamma^2(1+2 \alpha)}\left(1+\mathcal{O}\left(t^{2 \alpha+1}\right)\right), \\
v_k(t)=1+\mathcal{O}\left(t^{2 \alpha+1}\right),
\end{array}\right.
    \end{equation}
 with $c_k$ given in \eqref{defck}  for $k\neq m$.
Furthermore, the Hamiltonian $H(t)$ associated with the solutions has the asymptotic behavior \eqref{eq:Hasyzero}  as $t\to 0^+$.
\end{pro}

\begin{proof}
Tracing back the invertible  transformations $Y\to X\to E$ given in \eqref{X with phi} and \eqref{def:E}, we obtain
\begin{equation}\label{eq:YXE}
    Y(z/2t) = \frac{1}{2 t}(2 t)^{\beta \sigma_3} E(z)   \begin{cases} L(z), & |z|<\varepsilon, \\  \Phi(z), & |z|>\varepsilon,\end{cases}
\end{equation}
where $L(z)$, $\Phi(z)$ and $E(z)$ are given in \eqref{construct L}, \eqref{value of E}, \eqref{value of E2} and \eqref{solution4phi}.

Comparing the large-$z$ expansion on both sides of the second equation in \eqref{eq:YXE}, we have
\begin{equation}\label{eq:Y1Asy0}
    Y_1 = \frac{1}{2 t}(2 t)^{\beta \sigma_3} (E_1+\Phi_1)(2 t)^{-\beta \sigma_3},
\end{equation}
where $Y_1$, $E_1$ and $\Phi_1$ are the coefficients of $1/z$ in the large-$z$ expansion of $Y(z)$, $E(z)$ and $\Phi{(z)}$, respectively.
Substituting \eqref{value of E} and \eqref{phi aym inf} into \eqref{eq:Y1Asy0} yields
\begin{equation}\label{asym of Y1 as t goes to 0}
    {(Y_1)}_{11}= \sum_{k=0}^n \frac{c_k \Gamma(1+\alpha-\beta) \Gamma(1+\alpha+\beta) \left(2 |r_k|\right)^{2 \alpha+1}}{2(2 \alpha+1) \Gamma^2(1+2 \alpha)} t^{2 \alpha}+\frac{\alpha^2-\beta^2}{2i t}+\mathcal{O}\left(t^{2 \alpha+1}\right),
\end{equation}
as $t\to0^+$.
This, together with \eqref{defH}, implies the asymptotics of the Hamiltonian as given in \eqref{eq:Hasyzero}.

To derive the asymptotics of $y(t)$  and $d(t)$,  we consider the asymptotic expansion of $Y(z)$  as $t\to 0^+$ for $z$ near the origin.
Plugging  \eqref{Y asym at 0}, \eqref{construct L} and \eqref{value of E2} into the first equation in \eqref{eq:YXE}, we have
\begin{equation}\label{eq: expYm}
 Y_0^{(m)}\left[\begin{array}{l}
1 \\
0
\end{array}\right]=(2t)^{\alpha}(2 t)^{\beta \sigma_3}\left(I+\mathcal{O}\left(t^{2 \alpha+1}\right)\right) \Phi^{(0)}(0)\left[\begin{array}{l}
1 \\
0
\end{array}\right],
\end{equation}
where $ \Phi^{(0)}(0)$ is given in \eqref{eq: Phi0}. Using \eqref{solution4phi}, we have from \eqref{eq: expYm} that
\begin{equation}\label{Y_0^m}
    Y_0^{(m)}=\left[\begin{array}{ll}
(2 t)^{\alpha+\beta} e^{-\frac{\alpha+\beta}{2} \pi i} \frac{\Gamma(1+\alpha-\beta)}{\Gamma(1+2 \alpha)} & * \\
(2 t)^{\alpha-\beta} e^{-\frac{\alpha-\beta}{2} \pi i} \frac{\Gamma(1+\alpha+\beta)}{\Gamma(1+2 \alpha)} & *
\end{array}\right]\left(I+\mathcal{O}(t^{2\alpha+1})\right),
\end{equation}
for $\alpha>-1/2$, $2\alpha\not\in \mathbb{N}$ and $\beta\in i\mathbb{R}$. The case $2\alpha\in \mathbb{N}$ can be derived in a similar way and we skip the details.
Recalling  \eqref{y with rh} and \eqref{d with rh}, we derive the  asymptotics of $y(t)$ and $d(t)$ as $t\rightarrow0^+$:
\begin{equation}\label{y0}
    y(t)=\frac{\Gamma(1+\alpha-\beta)}{\Gamma(1+\alpha+\beta)} e^{-\beta \pi i}(2 t)^{2 \beta}\left(1+\mathcal{O}\left(t^{2 \alpha+1}\right)\right),
\end{equation}
\begin{equation} \label{d0}
    d(t)=2\alpha \frac{\Gamma(1+\alpha-\beta) \Gamma(1+\alpha+\beta)}{\Gamma^2(1+2 \alpha)} e^{-\alpha \pi i}(2 t)^{2 \alpha}\left(1+\mathcal{O}\left(t^{2 \alpha+1}\right)\right).
\end{equation}

Finally, we derive the asymptotics of $u_k$ and $v_k$  for $k\neq m$ by using \eqref{y0} and the following relations
\begin{equation}\label{eq:v_kexp}
    v_k(t)=\frac{\left(Y_0^{(k)}(t)\right)_{21}}{\left(Y_0^{(k)}(t)\right)_{11}} y(t),
\end{equation}
and
\begin{equation}\label{eq:u_kexp}
u_k(t)=\frac{\Tilde{c_k}\left(Y_0^{(k)}\right)_{11}\left(Y_0^{(k)}\right)_{21}}{v(t)},
\end{equation}
which can be seen from Proposition \ref{lax pair and cpv} and \eqref{Y asym at rk}. Similarly, we obtain after substituting  \eqref{construct L} and \eqref{value of E2} into the first equation in \eqref{eq:YXE} that
\begin{equation}\label{eq for Y_0^k}
     Y(z) =(2 t)^{\beta \sigma_3}\left(I+\mathcal{O}\left(t^{2 \alpha+1}\right)\right) \Phi^{(0)}(z)\left[\begin{array}{cc}1 & t^{2 \alpha} k\left(\frac{z}{t}\right) \\ 0 & 1\end{array}\right] z^{\alpha \sigma_3} \left[\begin{array}{cc}
1 & \frac{\sin (\alpha+\beta) \pi}{\sin 2 \alpha \pi} \\
0 & 1
\end{array}\right].
\end{equation}
From the integral expression of $k(z)$ in \eqref{expression of k}, we have
\begin{equation}\label{singularity at 2rk}
     k(z) \sim \operatorname{sgn}(r_k)c_ke^{\alpha\pi i}\left(2 |r_k|\right)^{2 \alpha} \ln \left(z-2 r_k\right),\end{equation}
     as $z \rightarrow 2 r_k $ for $k\neq m$.
 Comparing \eqref{Y asym at rk} with \eqref{eq for Y_0^k} and \eqref{singularity at 2rk}, we obtain
 \begin{equation}\label{eq:YoAsy}
  Y_0^{(k)}  =\left[\begin{array}{ll}(2 t)^{\alpha+\beta} e^{-\frac{\alpha+\beta}{2} \pi i} \frac{\Gamma(1+\alpha-\beta)}{\Gamma(1+2 \alpha)}r_k^\alpha & * \\ (2 t)^{\alpha-\beta} e^{-\frac{\alpha-\beta}{2} \pi i} \frac{\Gamma(1+\alpha+\beta)}{\Gamma(1+2 \alpha)}r_k^\alpha & *\end{array}\right](I+\mathcal{O}(t^{2\alpha+1}))
\end{equation}
for $r_k>0$,
and
 \begin{equation}\label{eq:YoAsy2}
  Y_0^{(k)}  =\left[\begin{array}{ll}(2 t)^{\alpha+\beta} e^{-\frac{\alpha+\beta}{2} \pi i} \frac{\Gamma(1+\alpha-\beta)}{\Gamma(1+2 \alpha)}\left(|r_k|e^{\pi i}\right)^\alpha & * \\ (2 t)^{\alpha-\beta} e^{-\frac{\alpha-\beta}{2} \pi i} \frac{\Gamma(1+\alpha+\beta)}{\Gamma(1+2 \alpha)}\left(|r_k|e^{\pi i}\right)^\alpha & *\end{array}\right](I+\mathcal{O}(t^{2\alpha+1}))
\end{equation}
for $r_k<0$.
Using \eqref{y0}, \eqref{eq:v_kexp}, \eqref{eq:u_kexp},  \eqref{eq:YoAsy} and \eqref{eq:YoAsy2}, we have  \eqref{CPVasym0}.
This completes the proof of the proposition.
\end{proof}

\section{Asymptotic Analysis as $t$ $\rightarrow$ $+\infty $}\label{sec:Asyinfty}
In this section, we carry out the Deift-Zhou nonlinear steepest descent analysis of the RH problem for $Y(z)$ as $t\rightarrow+\infty$.
Based on this, we obtain the  asymptotic behavior of a family of solutions to the coupled Painlev\'e V system and the associated Hamiltonian as $t\rightarrow+\infty$.

\subsection{Normalization}
In the first transformation, we normalize the behavior of $Y(z)$ at infinity.
\begin{equation}\label{eq:5.2}
    T(z)=Y(z, t) e^{i z t \sigma_3},
\end{equation}
Then, $T(z)$ satisfies the following RH problem.
\subsection*{RH problem for $T(z)$}
\begin{itemize}
    \item[\rm (1)] $T(z)$ is analytic in $\mathbb{C}\backslash\cup_{j=1}^7\Gamma_j$. The jump contours are shown in Figure \ref{fig:3}.

    \item[\rm (2)] $T(z)$ satisfies the jump condition
     \begin{equation}\label{eq:Tjump}
        T_{+}(z)=T_{-}(z)J_T(z),
    \end{equation}
   where    \begin{equation}
            J_T(z)= \begin{cases}{\left[\begin{array}{cc}
                              1 & 0 \\
                              e^{-(\alpha-\beta) \pi i} e^{2 i z t} & 1
                              \end{array}\right]},&{z\in \Gamma_1,} \\[.4cm]
                            {\left[\begin{array}{cc}
                              1 & 0 \\
                              e^{(\alpha-\beta) \pi i} e^{2 i z t} & 1
                             \end{array}\right],}&{z \in \Gamma_2,}\\[.4cm]
                            {\left[\begin{array}{cc}
1 & -e^{-(\alpha-\beta) \pi i} e^{-2 i z t} \\
0 & 1
\end{array}\right],}&{z\in \Gamma_3,}\\[.4cm]
{e^{2 \beta \pi i \sigma_3},}&{ z\in\Gamma_4,}\\[.4cm]
{\left[\begin{array}{cc}
1 & -e^{(\alpha-\beta) \pi i} e^{-2 i z t} \\
0 & 1
\end{array}\right],}&{ z\in\Gamma_5,}\\[.4cm]
{\left[\begin{array}{cc}
0 & -e^{-(\alpha-\beta) \pi i} e^{-2 i z t} \\
e^{(\alpha-\beta) \pi i} e^{2 i z t} & 1-\sigma(z t)
\end{array}\right],}&{z\in \Gamma_6,}\\[.4cm]
{\left[\begin{array}{cc}
0 & -e^{(\alpha-\beta) \pi i} e^{-2 i z t} \\
e^{-(\alpha-\beta) \pi i} e^{2 i z t} & 1-\sigma(z t)
\end{array}\right],}&{z\in \Gamma_7.}
\end{cases}
    \end{equation}

    \item[\rm (3)] As $z\rightarrow\infty$,
    \begin{equation}
        T(z)=\left(I+\frac{Y_1(t)}{z}+\mathcal{O}\left(\frac{1}{z^2}\right)\right) z^{-\beta \sigma_3}.
    \end{equation}

    \item[\rm (4)] The behavior of $T(z)$ near $r_k$, $k=0,\dots,n$, are the same as that of $Y(z)$, as given in (\ref{Y asym at rk}) and (\ref{Y asym at 0}).
\end{itemize}

\subsection{Opening  lenses and deformation}
It is readily observed that the off-diagonal entries in the jump matrices on $\Gamma_6$ and $\Gamma_7$ are highly oscillatory as $t\rightarrow+\infty$. To eliminate the oscillation, we factorize the jump matrices on $\Gamma_6$ and $\Gamma_7$ as follow:
\begin{equation}
   \begin{aligned}
     &\left[\begin{array}{cc}
0 & -e^{-(\alpha-\beta) \pi i} e^{-2 i z t} \\
e^{(\alpha-\beta) \pi i}e^{ 2 i z t} & 1-\sigma
\end{array}\right]=  \\
& \left[\begin{array}{cc}
1 & -\frac{1}{1-\sigma} e^{-(\alpha-\beta) \pi i} e^{-2 i z t} \\
0 & 1
\end{array}\right]\left[\begin{array}{cc}
\frac{1}{1-\sigma} & 0 \\
0 & 1-\sigma
\end{array}\right]
\left[\begin{array}{cc}
1 & 0 \\
\frac{1}{1-\sigma} e^{(\alpha-\beta) \pi i} e^{2 i z t} & 1
\end{array}\right], \quad z\in\Gamma_6,
\end{aligned}
\end{equation}

\begin{equation}
    \begin{aligned}
        &\left[\begin{array}{cc}
0 & -e^{(\alpha-\beta) \pi i} e^{-2 i z t} \\
e^{-(\alpha-\beta) \pi i} e^{2 i z t} & 1-\sigma
\end{array}\right]=\\
&\left[\begin{array}{cc}
1 & -\frac{1}{1-\sigma} e^{(\alpha-\beta) \pi i} e^{-2 i z t} \\
0 & 1
\end{array}\right]\left[\begin{array}{cc}
\frac{1}{1-\sigma} & 0 \\
0 & 1-\sigma
\end{array}\right]\left[\begin{array}{cc}
1 & 0 \\
\frac{1}{1-\sigma} e^{-(\alpha-\beta) \pi i}e^{2 i z t} & 1
\end{array}\right], \quad z\in\Gamma_7.
    \end{aligned}
\end{equation}
We also deform the contour $\Gamma_6$ and $\Gamma_7$ into lens-shape regions as shown in Figure \ref{fig:S}.
\begin{figure}[h]
    \centering
    \includegraphics[width=1.0\textwidth]{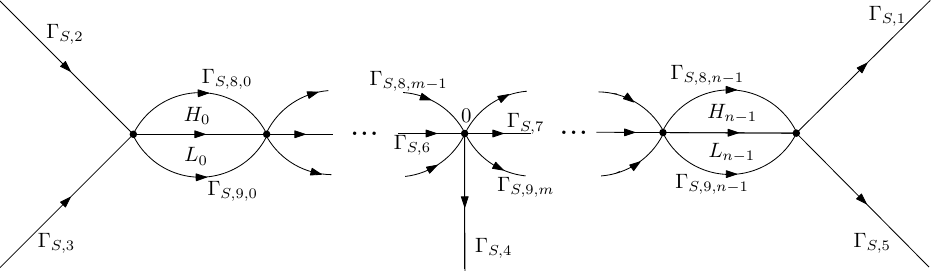}
    \caption{\small{Jump contours and regions for $S$}}
    \label{fig:S}
\end{figure}

Based on the factorization of jumps and deformation of contours,  we introduce the second transformation
\begin{equation}\label{eq:5.8}
    S(z)=T(z)\left\{\begin{array}{ll}
{\left[\begin{array}{cc}
1 & 0 \\
-\frac{1}{1-\sigma} e^{(\alpha-\beta) \pi i} e^{2 i z t} & 1
\end{array}\right]}, & {z\in H_k, 0\leq k<m,} \\
{\left[\begin{array}{cc}
1 & 0 \\
-\frac{1}{1-\sigma} e^{-(\alpha-\beta) \pi i} e^{2 i z t} & 1
\end{array}\right]}, & {z\in H_k, m\leq k < n,}\\
{\left[\begin{array}{cc}
1 & -\frac{1}{1-\sigma} e^{-(\alpha-\beta) \pi i} e^{-2 i z t} \\
0 & 1
\end{array}\right]},& {z\in L_k, 0\leq k<m,} \\
{\left[\begin{array}{cc}
1 & -\frac{1}{1-\sigma} e^{(\alpha-\beta) \pi i} e^{-2 i z t} \\
0 & 1
\end{array}\right]},&{z\in L_k, m\leq k < n,}\\
I & \text{ elsewhere. }
\end{array}\right.
\end{equation}
Then, $S(z)$  satisfies the following RH problem.
\subsection*{RH problem for $S(z)$}
\begin{itemize}
    \item[\rm (1)] $S(z)$ is analytic in $\mathbb{C}\backslash\Gamma_S$ as shown in Figure \ref{fig:S}.

    \item[\rm (2)] $S(z)$ satisfies the jump condition
     \begin{equation}\label{eq:Sjump}
        S_{+}(z)=S_{-}(z)J_S(z),
    \end{equation}
   where
    \begin{equation}\label{eq:5.9}
        J_S(z)=\left\{\begin{array}{ll}
J_T(z), &z\in \bigcup_{j=1}^5 \Gamma_{S,j}, \\
{\left[\begin{array}{cc}
1 & 0\\
\frac{1}{1-\gamma_k} e^{(\alpha-\beta) \pi i}e^{2 i z t} & 1
\end{array}\right]},& {\text{on } \Gamma_{S, 8 ,k}, k<m},\\
{\left[\begin{array}{cc}
1 & 0 \\
\frac{1}{1-\gamma_k} e^{-(\alpha-\beta) \pi i} e^{2 i z t} & 1
\end{array}\right]},&{\text{on } \Gamma_{S,8,k},k\geqslant m},\\
{\left[\begin{array}{cc}
1 & -\frac{1}{1-\gamma_k} e^{-(\alpha-\beta) \pi i} e^{-2 i z t} \\
0 & 1
\end{array}\right]}, &{\text{on } \Gamma_{S,9,k}, k<m},\\
{\left[\begin{array}{cc}
1 & -\frac{1}{1-\gamma_k} e^{(\alpha-\beta) \pi i} e^{-2 i z t} \\
0 & 1
\end{array}\right]},&{\text{on } \Gamma_{S,9,k}, k\geqslant m},\\
{\left[\begin{array}{cc}
\frac{1}{1-\sigma(z t)} & 0 \\
0 & 1-\sigma(z t)
\end{array}\right]},&{\text{on }\Gamma_{S,6}\cup \Gamma_{S,7}.}
\end{array}\right.
    \end{equation}
 \item[\rm (3)] As $z\to\infty$, we have
 \begin{equation}
        S(z)=\left(I+\frac{Y_1(t)}{z}+\mathcal{O}\left(\frac{1}{z^2}\right)\right) z^{-\beta \sigma_3}.
    \end{equation}
    \item[\rm (4)] The asymptotics of $S(z)$ are the same as $T(z)$ near $r_k$, $k=0,\dots,n$.
    \end{itemize}

\subsection{Global parametrix}
The jumps for $S(z)$ tend to the identity matrix exponentially fast except the ones on the real axis as $t\to\infty$. Therefore, we arrive at
the following approximate RH problem as $t\rightarrow+\infty$.
\subsection*{RH problem for $P^{(\infty)}(z)$}
\begin{itemize}
    \item[\rm (1)] $P^{(\infty)}(z)$ is analytic for $z\in \mathbb{C}\backslash \{[r_0,r_n]\cup \Gamma_{S,4}\}$.

    \item[\rm (2)] $P^{(\infty)}(z)$ satisfies the following jump condition
    \begin{equation}\label{jump for S inf}
        P_{+}^{(\infty)}(z)=P_{-}^{(\infty)}(z)\left\{\begin{array}{ll}
{\left[\begin{array}{cc}
\frac{1}{1-\sigma} & 0 \\
0 & 1-\sigma
\end{array}\right],} & z \in\left(r_0, r_n\right),\\[.4cm]
e^{2 \beta \pi i \sigma_3}, & z \in \Gamma_{S, 4},
\end{array}\right.
    \end{equation}

    \item[\rm (3)] As $z\rightarrow\infty$,
    \begin{equation}\label{asym for Sinf at inf}
        P^{(\infty)}(z)=\left(I+\mathcal{O}\left(\frac{1}{z}\right)\right) z^{-\beta \sigma_3}.
    \end{equation}
\end{itemize}
The solution $P^{(\infty)}(z)$ can be constructed explicitly as below
\begin{equation}\label{Sinf}
    P^{(\infty)}(z)=\prod_{k=0}^{n-1}\left(\frac{z-r_{k+1}}{z-r_k}\right)^{-\frac{\ln \left(1-\gamma_k\right)}{2 \pi i}\sigma_3}  z^{-\beta\sigma_3},
\end{equation}
where $\left(\frac{z-r_{k+1}}{z-r_k}\right)^{-\frac{\ln \left(1-\gamma_k\right)}{2 \pi i}}$  takes the principal branch for $0\leq k\leq n-1$ and the branch of $z^{-\beta}$ is taken such that arg$z\in(-\frac{\pi}{2},\frac{3\pi}{2})$.

\subsection{Local parametrices}

Near the points $r_k$, $k=0,\dots,n$, the solution of the RH problem $S(z)$ can not be approximated by the global parametrix $P^{(\infty)}(z)$. In this section,  we construct the local parametrices in the neighborhoods of $r_k$, $k=0,\dots,n$.

\subsection*{RH problem for $P^{(r_k)}(z)$}
\begin{itemize}
  \item[\rm (1)]  $P^{(r_k)}(z)$ is analytic for $z\in U(r_k, \varepsilon) \setminus \Gamma_S$ with some constant $\varepsilon>0$.
      \item[\rm (2)] $P^{(r_k)}(z)$ has the same jumps as $S(z)$ for $z \in \Gamma_S\cap U(r_k, \varepsilon) $.
      \item[\rm (3)] As $t\rightarrow +\infty$, we have
    \begin{equation}\label{SMatching}S(z)=\left(I+O\left(\frac{1}{t}\right)\right)P^{(r_k)}(z),
    \end{equation} where the error term is uniform for $z\in \partial U(r_k, \varepsilon)$.
\end{itemize}

For $k\neq m$, the solution can be constructed by using the confluent hypergeometric parametrix $\Phi^{(\mathrm{CHF})}$
 \begin{equation}\label{S^rk}
    \begin{aligned}
P^{\left(r_k\right)}(z)=& E^{\left(r_k\right)}(z) \Phi^{(\mathrm{CHF})}\left(2it(z-r_k) ; 0, b_k\right)\left[\left(1-\gamma_k\right)\left(1-\gamma_{k-1}\right)\right]^{-\frac{\sigma_3}{4}}\\
& \times \begin{cases}e^{-\operatorname{sgn}(r_k) \frac{\alpha-\beta}{2}\pi i \sigma_3} e^{i z t \sigma_3}, &  \operatorname{Im} z>0, \\
{\left[\begin{array}{cc}
0 & 1 \\
-1 & 0
\end{array}\right] e^{-\operatorname{sgn}(r_k)\frac{\alpha-\beta}{2} \pi i \sigma_3} e^{i z t \sigma_3}}, &  \operatorname{Im} z<0,\end{cases}
    \end{aligned}
\end{equation} where \begin{equation}\label{E^rk}
    \begin{aligned}
 E^{\left(r_k\right)}(z)=&\prod_{j \neq k}\left(z-r_j\right)^{-b_j\sigma_3} z^{-\beta \sigma_3}(2 i t)^{b_k \sigma_3} e^{-i r_k t \sigma_3} e^{\operatorname{sgn}(r_k)\frac{\alpha-\beta}{2}\pi i \sigma_3}\\
& \times \begin{cases}\left(1-\gamma_{k-1}\right)^{-\frac{\sigma_3}{4}}\left(1-\gamma_k\right)^{\frac{3}{4} \sigma_3}, & \operatorname{Im} z>0, \\
{\left[\left(1-\gamma_k\right)\left(1-\gamma_{k-1}\right)\right]^{-\frac{\sigma_3}{4}},} & \operatorname{Im} z<0,\end{cases} \\
&
\end{aligned}
\end{equation}
with $\gamma_{-1}=\gamma_n=0$.
While for $r_m = 0$, we construct $P^{(r_m)}$ as follows
\begin{equation}\label{S^0}
    \begin{aligned}
P^{\left(r_m\right)}(z)= & E^{\left(r_m\right)}(z) \Phi^{(\mathrm{CHF})}\left(2it z; \alpha, \beta+b_m\right)\left[\left(1-\gamma_{m-1}\right)\left(1-\gamma_m\right)\right]^{-\frac{\sigma_3}{4}} \\
& \times\begin{cases}e^{\frac{\alpha}{2} \pi i \sigma_3} e^{i z t \sigma_3}, & \arg z \in\left(0, \frac{\pi}{2}\right), \\
e^{-\frac{\alpha}{2} \pi i \sigma_3} e^{i z t \sigma_3}, & \arg z \in\left(\frac{\pi}{2}, \pi\right),\\
e^{\frac{\alpha+2 \beta}{2} \pi i \sigma_3}\left[\begin{array}{cc}
0 & 1 \\
-1 & 0
\end{array}\right] e^{i z t \sigma_3}, &  \arg z \in\left(-\pi, \frac{\pi}{2}\right), \\
e^{-\frac{\alpha+2 \beta}{2} \pi i \sigma_3}\left[\begin{array}{cc}
0 & 1 \\
-1 & 0
\end{array}\right] e^{i z t \sigma_3}, &  \arg z \in\left(-\frac{\pi}{2}, 0\right),\end{cases}
\end{aligned}
\end{equation}
where
\begin{equation}\label{E^0}
    \begin{aligned}
E^{\left(r_m\right)}(z)= &\prod_{j \neq m}\left(z-r_j\right)^{-b_j\sigma_3}(2 i t)^{\left(\beta+b_m\right) \sigma_3} \\
 &\times\begin{cases}e^{-\beta \pi i \sigma_3}\left(1-\gamma_{m-1}\right)^{-\frac{\sigma_3}{4}}\left(1-\gamma_m\right)^{\frac{3}{4} \sigma_3}, &  \arg z\in(0,\pi), \\
e^{\beta \pi i \sigma_3}\left[\left(1-\gamma_{m-1}\right)\left(1-\gamma_m\right)\right]^{-\frac{\sigma_3}{4}}, & \arg  z \in\left(-\pi,-\frac{\pi}{2}\right), \\
e^{-\beta \pi i \sigma_3}\left[\left(1-\gamma_{m-1}\right)\left(1-\gamma_m\right)\right]^{-\frac{\sigma_3}{4}},& \arg  z \in\left(-\frac{\pi}{2}, 0\right).\end{cases} \\
&
\end{aligned}
\end{equation}
The branch for  $z^{\beta}$ is taken such that  $\arg z \in (-\frac{\pi}{2}, \frac{3\pi}{2})$. While the  branches of the other power functions in  \eqref{S^rk}, \eqref{E^rk}, \eqref{S^0} and \eqref{E^0}  take the principal branches.

\subsection{Final transformation}
In the final transformation, we define
\begin{equation}\label{def R(eta)}
    R(z)= \begin{cases}{S(z)P^{\left(r_k\right)}(z)}^{-1}, &  z\in U\left(r_k, \varepsilon\right),~~0\leq k\leq n,\\S(z) {P^{(\infty)}(z)}^{-1} & \text { elsewhere.}\end{cases}
\end{equation}
\subsection*{RH problem for $R(z)$}
\begin{itemize}
   \item[\rm (1)]The jump contours for $R$ are shown in Figure \ref{fig:R}
    \begin{figure}[h]
        \centering
       \includegraphics[width=0.95\textwidth]{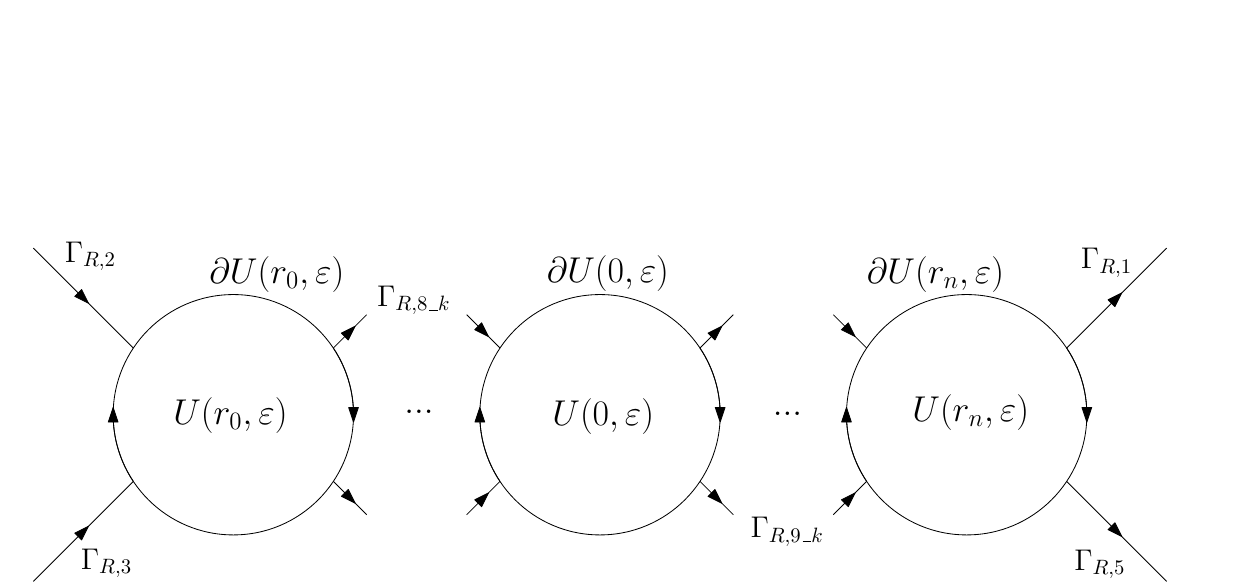}
        \caption{\small{Jump contours
$\Gamma_R$ and regions for $R$}}
        \label{fig:R}
    \end{figure}

    \item[\rm (2)] $R(z)$ satisfies the jump condition
     \begin{equation}\label{eq:Rjump}
        R_{+}(z)=R_{-}(z)J_R(z),
    \end{equation}
where
     \begin{equation}\label{def:J_R}
        J_R(z)= \begin{cases}
        {P^{(r_k)}(z) {P^{(\infty)}(z)}^{-1}}, &  z\in\partial U(r_k, \varepsilon), ~~0\leq k\leq n, \\
        P^{(\infty)}(z) J_S(z) {P^{(\infty)}(z)}^{-1} &  \text {elsewhere}.\end{cases}
    \end{equation}

  \item[\rm (3)] As $z \rightarrow \infty$,
    \begin{equation}
        R(z)=I+\mathcal{O}(\frac{1}{z}).
    \end{equation}
\end{itemize}
From \eqref{eq:5.9} and \eqref{SMatching}, we have
 \begin{equation}\label{eq:RjumpEst}
       J_R(z)=I+\mathcal{O}\left(\frac{1}{t}\right),
    \end{equation}
where the error term is uniform for $z $ on the jump contours for $R(z)$, $\gamma_j$ in any compact subset of $[0,1)$, and $r_0<\dots<r_{m}=0<\dots<r_n$ in any compact subset of $\mathbb{R}$ and bounded away from each other.
Therefore, we have as $t\rightarrow+\infty$
 \begin{equation}\label{expansion of R1}
    R(z)=I+\mathcal{O}\left(\frac{1}{t}\right),
\end{equation}
where the error term is uniform for $z $ in the complex plane.
For later use, we derive from  \eqref{Sinf}, \eqref{S^rk}, \eqref{E^rk} and \eqref{phi aym inf} the following expansion
\begin{equation}\label{expansion of JR by t}
    J_R(z)=I+\frac{J_{R, 1}(z)}{t}+\mathcal{O}\left(\frac{1}{t^2}\right),
\end{equation}
 where
 \begin{equation}\label{JR,1 k}
    J_{R, 1}(z)=-\frac{b_k^2}{2 i\left(z-r_k\right)} E^{\left(r_k\right)}(z)\left[\begin{array}{cc}
1 & \frac{\Gamma\left(-b_k\right)}{\Gamma\left(1+b_k\right)} \\
\frac{-\Gamma\left(b_k\right)}{\Gamma\left(1-b_k\right)} & -1
\end{array}\right] {E^{\left(r_k\right)}(z)}^{-1},
\end{equation}
for $z\in\partial U(r_k, \varepsilon)$, $k\neq m$ and
 \begin{equation}\label{JR,1 0}
    J_{R, 1}(z)=\frac{\alpha^2-\left(\beta+b_m\right)^2}{2 i z} E^{\left(r_m\right)}(z)\left[\begin{array}{cc}
1 & \frac{\Gamma\left(\alpha-\beta-b_m\right)}{\Gamma\left(1+\alpha+\beta+b_m\right)} \\
\frac{-\Gamma\left(\alpha+\beta+b_m\right)}{\Gamma\left(1+\alpha-\beta-b_m\right)} & -1
\end{array}\right]{E^{(r_m)}(z)}^{-1},
\end{equation}
for $z\in\partial U(0, \varepsilon)$.
 Therefore, we have as $t\rightarrow+\infty$
 \begin{equation}\label{expansion of R by t}
    R(z)=I+\frac{R_1(z)}{t}+\mathcal{O}\left(\frac{1}{t^2}\right),
\end{equation}
 where
 \begin{equation}\label{integral expression of R1}
    R_1(z)=\frac{1}{2 \pi i} \sum_{k=0}^n \int_{\partial U\left(r_k, \xi\right)} \frac{J_{R, 1}(x)}{x-z} d x.
\end{equation}
After some residue  calculation, we have
  \begin{equation}\label{R1,11}
    \left(R_1(z)\right)_{11}= \begin{cases}-\sum_{j=0}^n \frac{b_j^2}{2 i\left(z-r_j\right)}+\frac{\alpha^2-\beta^2-2 \beta b_m}{2 i z}, &  z \in \mathbb{C} \backslash \bigcup^n_{j=0}{U(r_j,\varepsilon),} \\
     -\sum_{j \neq k} \frac{b_j^2}{2 i\left(z-r_j\right)}+\frac{\alpha^2-\beta^2-2 \beta b_m}{2 i z},&  z \in U\left(r_k, \varepsilon\right),~  k\neq m, \\ -\sum_{j \neq m} \frac{b_j^2}{2 i\left(z-r_j\right)}, &  z \in U\left(r_m, \varepsilon\right).\end{cases}
\end{equation}

{For later use, we derive the following estimate of derivatives of $J_R(z)$ with respect to $\gamma_j$ as $t\rightarrow +\infty$,}
\begin{equation}\label{estimateJr}
    {\left\|\left(\prod_{j=0}^{n-1}\partial_{\gamma_j}^{k_j}\right)(J_R(z)-I)\right\|_{L^2(\Gamma_R)\cap L^{\infty}(\Gamma_R)}\leq  c  \frac{(\ln t)^k}{t},}
\end{equation}
where $k=\sum_{j=0}^{n-1}k_j$ and the constant $c>0$.  The factor $(\ln t)^k$ is due to the factors $t^{\pm b_j}$ in the entries of $J_R(z)$, cf. \eqref{S^rk}-\eqref{E^0} and \eqref{def:J_R}. It is worth pointing out  that the estimate holds uniformly for $\gamma_j$ in any compact subset of $[0,1)$, and uniformly for $r_0<\dots<r_{m}=0<\dots<r_n$ in any compact subset of $\mathbb{R}$ and bounded away from each other.
 Then, by \eqref{integral expression of R1}, we obtain
\begin{equation}\label{derivativeR1}
    \left(\prod_{j=0}^{n-1}\partial_{\gamma_j}^{k_j}\right) R_1(z)=\mathcal{O}((\ln t)^k), \quad \text{as } t\rightarrow +\infty.
\end{equation}

\subsection{Asymptotic behaviors as $t \rightarrow +\infty$}
Using the results of the nonlinear steepest descent analysis of the RH problem for $Y(z)$, we obtain the asymptotics of  a family of special solutions to the coupled Painlev\'e V system and the Hamiltonian associated as $t \rightarrow +\infty$.

\begin{pro}\label{pro: CPVasyinfty}
 Under the same conditions as in Proposition \ref{pro: CPVasy0}, the family of solutions of  the coupled Painlev\'e V system  \eqref{eq:CPV} with the asymptotic behaviors as $t\to 0^+$ given in Proposition \ref{pro: CPVasy0} have the asymptotics behaviors as $t\rightarrow+\infty$
     \begin{equation}\label{CPVasyminf}
          \left\{\begin{array}{l}
\begin{aligned}

u_k(t)= & \operatorname{sgn}(r_k)~c_k \cdot \frac{\Gamma^2\left(1-b_k\right) \Gamma(1+\alpha+\beta+b_m)}{\Gamma\left(1+\alpha-\beta-b_m\right)} \prod_{j \neq m , k}\left(\frac{r_k-r_j}{r_m-r_j}\right)_{+}^{-2 b_j}|r_k|^{2\left(b_k-b_m-\beta\right)} \\
&\left[\left(1-\gamma_{m-1}\right)\left(1-\gamma_m\right)\right]^{-\frac{1}{2}} e^{\operatorname{sgn}(r_k)\pi i\left(b_k+b_m+\alpha+\beta\right)}(2 t)^{2\left(b_k-b_m-\beta\right)}e^{-2 i t r_k}(1+\mathcal{O}(t^{-1})),
\end{aligned}\\
\begin{aligned}
v_k(t)=& \operatorname{sgn}(r_k)~ \frac{(\gamma_k-\gamma_{k-1}) }{2 \pi i c_k }\frac{\Gamma\left(1+\alpha-\beta-b_m\right) \Gamma\left(1+b_k\right)}{\Gamma\left(1+\alpha+\beta+b_m\right) \Gamma\left(1-b_k\right)} \prod_{j \neq m, k}\left(\frac{r_k-r_j}{r_m-r_j}\right)_{+}^{2 b_j}
|r_k|^{2\left(\beta+b_m- b_k\right)} \\
&
\left[\frac{\left(1-\gamma_{m-1}\right)\left(1-\gamma_m\right)}{\left(1-\gamma_{k-1}\right)\left(1-\gamma_k\right)}\right]^{\frac{1}{2}} e^{-\operatorname{sgn}(r_k)\pi i\left(b_k+b_m+\alpha+\beta\right)}(2 t)^{2\left(\beta+b_m-b_k\right)} e^{2 i t r_k}\left(1+\mathcal{O}\left(t^{-1}\right)\right),
\end{aligned}
\end{array}\right.
     \end{equation}
     for $k\neq m$ and  $c_k$ is defined in (\ref{defck}). Here the power functions take the principal branches.
Moreover, the Hamiltonian associated with this family of solutions has the asymptotics behavior as $t\rightarrow+\infty$
\begin{equation}\label{Hasyminf}
    H(t)=\sum_{k=0}^n 2 i b_k r_k-\frac{\sum_{k=0}^n b_k^2+2 \beta b_m}{t}+\mathcal{O}\left(t^{-2}\right).
\end{equation}
\end{pro}

\begin{proof}

Tracing back the transformations $Y \longmapsto N \longmapsto S \longmapsto R$,  we can derive the asymptotics of $Y(z)$ uniformly for $z$ in several different regions covering   the whole complex plane as $t\rightarrow+\infty$.

As $t\rightarrow +\infty$, we obtain from  (\ref{eq:5.2}), (\ref{eq:5.8}), (\ref{Sinf}), (\ref{def R(eta)}) and (\ref{expansion of R by t}) that \begin{equation}\label{Y with RS at inf}
    \begin{aligned}
Y(z) & =R(z) P^{(\infty)}(z) e^{-i z t \sigma_3} \\
& =\left(I+\frac{R_1(z)}{t}+\mathcal{O}\left(\frac{1}{t^2}\right)\right) \prod_{k=0}^{n-1}\left(\frac{z-r_{k+1}}{z-r_k}\right)^{-\frac{\ln \left(1-\gamma_k\right)}{2 \pi i} \sigma_3}z^{-\beta \sigma_3} e^{-i z t \sigma_3},
\end{aligned}
\end{equation}
where the error term is uniform for $|z-r_k|>\varepsilon$, $k=1,\dots,n$. Here, the branch for  $z^{\beta}$ is taken such that  $\arg z \in (-\frac{\pi}{2}, \frac{3\pi}{2})$ and the other power functions take the principal branches. Using(\ref{Y asym at inf}),  (\ref{R1,11}) and (\ref{Y with RS at inf}), we have
\begin{equation}\label{asym of Y1,11}
    \left(Y_1\right)_{11}=-\frac{1}{2 \pi i} \sum_{k=0}^{n-1}\left(r_k-r_{k+1}\right) \ln \left(1-\gamma_k\right)+\frac{\alpha^2-\beta^2-2 \beta b_m-\sum_{k=0}^n b_k^2}{2 i t} +\mathcal{O}\left(\frac{1}{t^2}\right).
\end{equation}
Substitute (\ref{asym of Y1,11}) into \eqref{defH}, we obtain  (\ref{Hasyminf}).

From (\ref{eq:5.2}), (\ref{eq:5.8}), (\ref{S^0}), (\ref{def R(eta)}) and (\ref{expansion of R by t}), we derive as $t\rightarrow+\infty$
\begin{equation}\label{eq:YExp0}
    \begin{aligned}
Y(z)= & R(z) P^{\left(r_m\right)}(z) e^{-i z t \sigma_3} \\
= & \left(I+\frac{R_1(z)}{t}+\mathcal{O}\left(\frac{1}{t^2})\right)\right) E^{\left(r_m\right)}(z) \Phi\left(2itz ; \alpha, \beta+b_m\right) \\
& {\left[\left(1-\gamma_{m-1}\right)\left(1-\gamma_m\right)\right]^{-\frac{\sigma_3}{4}} e^{\frac{\alpha}{2} \pi i \sigma_3}, }
\end{aligned}
\end{equation}
where $\arg z \in\left(0, \frac{\pi}{2}\right)$ and the error term is uniform for $|z|< \varepsilon$.
Using  (\ref{Y asym at 0}), (\ref{E^0}) and \eqref{R1,11}, we obtain from  \eqref{eq:YExp0} that
\begin{equation}\label{asym of Y_0^m}
    \begin{aligned}
 \left(Y_0^{(m)}\right)_{11}=&\frac{\Gamma\left(1+\alpha-\beta-b_m\right)}{\Gamma(1+2 \alpha)} \prod_{j \neq m}\left(-r_j\right)_{+}^{-b_j}e^{-i \frac{\pi}{2}(\alpha+\beta+b_m) }(2 t)^{\alpha+\beta+b_m}\left(1+\mathcal{O}\left(\frac{1}{t}\right)\right), \\
 \left(Y_0^{(m)}\right)_{21}=&\frac{\Gamma\left(1+\alpha+\beta+b_m\right)}{\Gamma(1+2 \alpha)} \prod_{j \neq m}\left(-r_j\right)_{+}^{b_j}e^{-i\frac{\pi}{2}(\alpha-\beta-b_m)}(2  t)^{\alpha-\beta-b_m}\\
 &\left[\left(1-\gamma_{m-1}\right)\left(1-\gamma_m\right)\right]^{-\frac{1}{2}}
 \left(1+\mathcal{O}\left(\frac{1}{t}\right)\right).
\end{aligned}
\end{equation}
Therefore, we have from  \eqref{y with rh}, \eqref{d with rh} and \eqref{asym of Y_0^m} that as $t\rightarrow+\infty$
\begin{equation}\label{y inf}
y(t)=\frac{\Gamma\left(1+\alpha-\beta-b_m\right)}{\Gamma\left(1+\alpha+\beta+b_m\right)} \prod_{j \neq m}\left(-r_j\right)_{+}^{-2 b_j}e^{-(\beta+b_m)\pi i}(2 t)^{2\left(\beta+b_m\right)}{\left[\left(1-\gamma_{m-1}\right)\left(1-\gamma_m\right)\right]^{\frac{1}{2}}\left(1+\mathcal{O}\left(\frac{1}{t}\right)\right),}
\end{equation}
and
\begin{equation}\label{d inf}
d(t)= 2\alpha \frac{\Gamma\left(1+\alpha-\beta-b_m\right) \Gamma\left(1+\alpha+\beta+b_m\right)}{\Gamma^2(1+2 \alpha)}e^{-\alpha \pi i}(2t)^{2 \alpha}\left[\left(1-\gamma_m\right)\left(1-\gamma_{m-1}\right)\right]^{-\frac{1}{2}}
\left(1+\mathcal{O}\left(
\frac{1}{t}\right)\right).
\end{equation}

Similarly, we have from  (\ref{eq:5.2}), (\ref{eq:5.8}) and (\ref{def R(eta)})  that as $t\rightarrow+\infty$
\begin{equation}\label{eq:5.41}
    Y(z)=R(z) P^{\left(r_k\right)}(z) e^{-i z t \sigma_3},
\end{equation}
where the error term is uniform for $|z-r_k|< \varepsilon$. From  \eqref{eq:5.2}, \eqref{eq:5.8}, \eqref{S^0}, \eqref{E^0} and \eqref{expansion of R by t}), we obtain
\begin{equation}\label{asym of Y_0^k}
    \begin{aligned}
\left(Y_0^{(k)}(t)\right)_{11}= & \Gamma\left(1-b_k\right) \prod_{j \neq k}\left(r_k-r_j\right)_{+}^{-b_j} r_k^{-\beta} e^{-\frac{b_k}{2}\pi i}(2t)^{b_k} e^{-i t r_k} (1+\mathcal{O}(\frac{1}{t})),\\
\left(Y_0^{(k)}(t)\right)_{21}= & \Gamma\left(1+b_k\right) \prod_{j \neq k}\left(r_k-r_j\right)_{+}^{b_j}  r_k^{\beta} e^{\frac{b_k}{2}\pi i}(2  t)^{-b_k} e^{-\operatorname{sgn}\left(r_k\right)(\alpha-\beta) \pi i}\\
&\left[\left(1-\gamma_{k-1}\right)\left(1-\gamma_k\right)\right]^{-\frac{1}{2}}(1+\mathcal{O}(\frac{1}{t})).
\end{aligned}
\end{equation}
From \eqref{eq:v_kexp}, \eqref{eq:u_kexp} and \eqref{asym of Y_0^k},  we have (\ref{CPVasyminf}). This completes the proof of the proposition.
\end{proof}

\section{Proof of main theorems}\label{sec:proofs}
\subsection{Proof of Theorem \ref{integralexpression}}

From  (\ref{fred with resolvent}), (\ref{defres}), (\ref{F with f}), (\ref{def model RH}) and (\ref{Y asym at rk}), we have
\begin{equation}\label{ie2}
\frac{d}{d t} \ln F=\frac{1}{t} \sum_{k=0}^n \Tilde{c}_k r_k\left(Y_1^{(k)}\right)_{21}, \end{equation}
where  $r_m=0$, $Y_1^{(k)}$ is the coefficient appears in the expansion \eqref{Y asym at rk} for $k=0,\dots,n$.  Here the coefficients $\Tilde{c}_k=c_ke^{\operatorname{sgn}(r_k)(\alpha+1) \pi i}$ and $c_k$ is defined in (\ref{defck}).
By   \eqref{Y asym at rk} and  \eqref{lax pair}, we have
\begin{equation}\label{rh with cpv 1}
    \left(\frac{d}{d t} Y_1^{(k)}\right)_{21}=2 i\left(Y_0^{(k)}\right)_{11}\left(Y_0^{(k)}\right)_{21}=\frac{2 i u_k v_k}{\Tilde{c}_k}.
\end{equation}
Therefore, we have
 \begin{equation}\label{ddF}
\frac{d}{d t} \left(t\frac{d}{d t} \ln F\right)=2 i \sum_{k = 1}^n r_ku_k v_k .\end{equation}
On the other side, we get from \eqref{eq:hamilton},  \eqref{defhamilton} \eqref{hamiltonform}  that
\begin{equation}\label{dH}
\frac{d}{d t}\left(tH(t)\right)=2 i \sum_{k = 1}^n r_ku_k v_k,\end{equation}
with $r_m=0$. From \eqref{eq:Hasyzero}, \eqref{ddF} and \eqref{dH}, we have
\begin{equation}
    \frac{d}{d t} \ln F=H(t).
\end{equation}
 By  \eqref{eq:Hasyzero},  the Hamiltonian  function $H$ is integrable near zero.
By using the fact that  $F(0)=1$, we obtain  the integral expression (\ref{eq:ie}) by taking integration on both sides of the above equation. This completes the proof of the theorem. \hfill$\square$

\subsection{Proof of Theorem \ref{thm large asymp}}
From the asymptotic behavior of the Hamiltonian function  \eqref{hamiltonsym0}, we obtain the asymptotics of the integral except the constant term
as $t\to+\infty$
\begin{equation}\label{eq:6.8}
    \int_0^t H(s) d s=\left(\sum_{k=0}^n 2 i b_k r_k\right) t-\left(\sum_{k=0}^n b_k^2+2 \beta b_m\right) \ln t+C_1+\mathcal{O}\left(\frac{1}{t}\right),
\end{equation}
where $C_1 = C_1(\alpha, \beta, \Vec{\gamma})$ is some constant depending on the parameters.

To determine the constant $C_1$, we integrate on both sides of (\ref{dift})
\begin{equation}\label{eq:6.9}
    \int_0^t H(s) d s=\int_0^t \sum_{k \neq m} u_k \frac{d v_k}{d s} -H(s)d s+\left.(s H(s)+\alpha \ln d(s)-\beta \ln y(s)-2(\alpha^2-\beta^2)\ln s)\right|_0 ^t.
\end{equation}
Substituting the asymptotics \eqref{eq:Hasyzero}, (\ref{y0}), (\ref{d0}), (\ref{Hasyminf}), (\ref{y inf}) and (\ref{d inf}) into the second term on the right-hand side of the above equation, we have
 \begin{equation}\label{asym rhs second part }
    \begin{aligned}
&\left.\left(s H+\alpha \ln d-\beta \ln y-2\left(\alpha^2-\beta^2\right) \ln s\right)\right|_0 ^t\\
&= \sum_{k=0}^n 2 i b_k r_k t-\sum_{k=0}^n b_k^2-2 \beta b_m+\alpha \ln \frac{\Gamma\left(1+\alpha-\beta-b_m\right) \Gamma\left(1+\alpha+\beta+b_m\right)}{\Gamma(1+\alpha-\beta) \Gamma(1+\alpha+\beta)} \\
&~~ -\beta \ln \frac{\Gamma\left(1+\alpha-\beta-b_m\right) \Gamma(1+\alpha+\beta)}{\Gamma\left(1+\alpha+\beta+b_m\right) \Gamma(1+\alpha+\beta)}+\beta \sum_{j \neq m} 2 b_j \ln |r_j|-2 \beta b_m \ln (2 t)\\
&~~-\frac{\alpha}{2} \ln \left[\left(1-\gamma_{m-1}\right)\left(1-\gamma_m\right)\right]+\mathcal{O}\left(\frac{1}{t}\right).
\end{aligned}
\end{equation}
Comparing \eqref{eq:6.8} with \eqref{eq:6.9} and \eqref{asym rhs second part }, we have
 \begin{equation}\label{asym rhs first part 2}
\int_0^t \sum_{k \neq m} u_k \frac{d v_k}{d s}-H d s=-\sum_{k=0}^n b_k^2 \ln t +C_2+\mathcal{O}\left(\frac{1}{t}\right),
\end{equation}
with\begin{equation}\label{C1 with C2}
    \begin{aligned}
C_1= & C_2+\alpha \ln \frac{\Gamma\left(1+\alpha-\beta-b_m\right) \Gamma\left(1+\alpha+\beta+b_m\right)}{\Gamma(1+\alpha-\beta) \Gamma(1+\alpha+\beta)}-\beta \ln \frac{\Gamma\left(1+\alpha-\beta-b_m\right) \Gamma(1+\alpha +\beta)}{\Gamma\left(1+\alpha+\beta+b_m\right) \Gamma(1+\alpha-\beta)} \\
& +\beta \sum_{k \neq m} 2 b_k \ln |r_k|-2 \beta b_m \ln (2)-\frac{\alpha}{2} \ln \left[\left(1-\gamma_{m-1}\right)\left(1-\gamma_m\right)\right]-\Sigma b^2_k-2\beta b_m.
\end{aligned}
\end{equation}
Therefore, the remaining  task is to derive the constant $C_2$. It is  more convenient to evaluate the integral constant $C_2$ since the derivatives of the integrand in \eqref{asym rhs first part 2} with respect to the parameters $\Vec{\gamma}$ is totally differential in $t$ as shown in \eqref{difgamma}.

Actually,  we obtain from  \eqref{difgamma}  that
 \begin{equation}\label{eq:db}
  \frac{d}{d b_j} C_2=\lim_{t\to+\infty} \left(\sum_{k\neq m} u_k(t) \frac{d v_k(t)}{d b_j}+2 b_j \ln t -2b_m \ln t\right)-\lim_{t\to 0^+} \left(\sum_{k\neq m} u_k(t) \frac{d v_k(t)}{d b_j}+2 b_j \ln t - 2b_m \ln t\right).
 \end{equation}
Substituting  into the above formula the asymptotics
 (\ref{CPVasym0}), (\ref{d0}), (\ref{CPVasyminf}) and (\ref{d inf}), we have
 \begin{equation}\label{eq:dC2}
\begin{aligned}\frac{d}{d b_j}C_2
=& b_j\left(\partial_{b_j} \ln \frac{\Gamma\left(1+\alpha-\beta-b_m\right)}{\Gamma\left(1+\alpha+\beta+b_m\right)}+\partial_{b_j} \ln \frac{\Gamma\left(1+b_j\right)}{\Gamma\left(1-b_j\right)}-4 \ln \left|2 r_j\right|\right) \\
&+\sum_{k \neq m, j} b_k\left(\partial_{b_j} \ln \frac{\Gamma\left(1+\alpha-\beta-b_m\right)}{\Gamma\left(1+\alpha+\beta+b_m\right)}-2 \ln \left|\frac{2 r_j r_k}{r_j-r_k}\right|\right),\\
\end{aligned}
\end{equation}
with $j\neq m$ and $b_m=-\sum_{j\neq m}b_j $.

When the parameters $\Vec{b}=\Vec{0} $, we have $\ln F(s)=0$ and the constant $C_1(0,\dots,0;\alpha,\beta)=C_2(0,\dots,0;\alpha,\beta)=0$. Using this initial value and the differential identity \eqref{eq:dC2}, we evaluate the constant $C_2$ by taking integration  with respect to the variables $b_j$, $j\neq m$ successively.  For the purpose,
we recall the  integral representation for the Barnes G-function (see \cite{NIST:DLMF})
\begin{equation}\label{def barnesG}
    \ln G(z+1)=\frac{1}{2} z \ln (2 \pi)-\frac{1}{2} z(z+1)+z \ln \Gamma(z+1)-\int_0^z \ln \Gamma(t+1) \mathrm{d} t.
\end{equation}
Using \eqref{def barnesG},  we  obtain by integrating  by parts that
 \begin{equation}\label{integral 2}
 \begin{aligned}
& \int_0^{b_j} x \partial_x \ln \frac{\Gamma(1+x)}{\Gamma(1-x)} d x =  b_j^2+\ln \left[G\left(1+b_j\right) G\left(1-b_j\right)\right].
\end{aligned}
\end{equation}
In the $j$-th step, taking $b_{j+1}= \dots =b_n=0$, we have $b_m=-s_j$ with $s_j=\sum_{l=0,l\neq m}^{j}b_l$. Then, we have from  integration by parts  and \eqref{def barnesG} that
\begin{equation}\label{integral 1}
\begin{aligned}
  & \int_0^{b_j}(s_j+x) \partial_x \ln \frac{\Gamma\left(1+\alpha-\beta+s_{j-1}+x\right)}{\Gamma\left(1+\alpha+\beta-s_{j-1}-x\right)} d x \\
  &=\left[x \ln \frac{\Gamma(1+\alpha-\beta+x)}{\Gamma(1+\alpha+\beta-x)} \right]_{s_{j-1}}^{s_j}-\int_{s_{j-1}}^{s_j}  \partial_x \ln \frac{\Gamma(1+\alpha-\beta+x)}{\Gamma(1+\alpha+\beta-x)} d x\\
&= b_j^2+2 b_js_{j-1}-(\alpha+\beta)\ln \frac{\Gamma(1+\alpha+\beta-s_{j}) }{\Gamma\left(1+\alpha+\beta-s_{j-1}\right) }
-(\alpha-\beta)\ln \frac{\Gamma(1+\alpha-\beta+s_{j}) }{\Gamma\left(1+\alpha-\beta+s_{j-1}\right) }\\
&~~+\ln \frac{G\left(1+\alpha-\beta+s_j\right) G\left(1+\alpha+\beta-s_j\right)}{G(1+\alpha-\beta+s_{j-1}) G(1+\alpha+\beta-s_{j-1})} .
\end{aligned}
\end{equation}

By \eqref{eq:dC2}, (\ref{integral 1}), (\ref{integral 2}) and the initial value $C_2(0,\dots,0;\alpha,\beta)=0$, we derive
\begin{equation}\label{C(b,alpha,beta)}
    \begin{aligned}
& C_2(\vec{b} ; \alpha, \beta)=\sum_{k\neq m} b_k^2+2 \beta b_m+\alpha \ln \frac{\Gamma(1+\alpha+\beta) \Gamma(1+\alpha-\beta)}{\Gamma\left(1+\alpha+\beta+b_m\right) \Gamma\left(1+\alpha-\beta-b_m\right)} \\
&+ \beta \ln \frac{\Gamma(1+\alpha+\beta) \Gamma\left(1+\alpha-\beta-b_m\right)}{\Gamma(1+\alpha-\beta) \Gamma\left(1+\alpha+\beta+b_m\right)}+\sum_{j \neq m} \ln \left[G\left(1+b_j\right) G\left(1-b_j\right)\right]\\
& +\ln \frac{G\left(1+\alpha+\beta+b_m\right) G\left(1+\alpha-\beta-b_m\right)}{G(1+\alpha+\beta) G(1+\alpha-\beta)} -\sum_{\substack{0 \leq j \neq k \leq n \\
j, k \neq m}} 2 b_j b_k \ln \left|\frac{2 r_j r_k}{r_j-r_k}\right|\\
& -\sum_{j \neq m} 2 b_j^2 \ln \left|2 r_j\right|.
\end{aligned}
\end{equation}
This, together with \eqref{C1 with C2}, implies
\begin{equation}\label{C1}
    \begin{aligned}
C_1= & -\sum_{k \neq m} 2 b_k^2 \ln \left|2 r_k\right|+\sum_{k \neq m} 2 b_k \beta \ln \left|2 r_k\right|-\sum_{\substack{0 \leq j<k \leqslant n \\
j, k \neq m}} 2 b_j b_k \ln \left|\frac{2 r_j r_k}{r_j-r_k}\right| \\
& +\sum_{k \neq m} \ln \left[G\left(1+b_j\right) G\left(1-b_j\right)\right]+\ln \frac{G\left(1+\alpha+\beta+b_m\right) G\left(1+\alpha-\beta-b_m\right)}{G(1+\alpha+\beta) G(1+\alpha-\beta)} \\
& -\frac{\alpha}{2} \ln \left[\left(1-\gamma_{m-1}\right)\left(1-\gamma_m\right)\right].
\end{aligned}
\end{equation}
Substitute (\ref{C1}) into (\ref{eq:6.8}), we obtain \eqref{large gap asym}.

{Finally, from \eqref{eq:ie}, \eqref{defH}, \eqref{difgamma}, \eqref{derivativeR1} and \eqref{Y with RS at inf}, we see that \eqref{large gap asym} can be be differentiated with respect to $\gamma_0, \dots, \gamma_{n-1}$ at the cost of increasing the error term by a factor $\ln t$ for each differentiation.} This completes the proof of the theorem.

\begin{appendices}

\section{The confluent hypergeometric  parametrix}\label{PCP}
\subsection*{RH problem for $\Phi^{(\mathrm{CHF})}$}
 \begin{description}
  \item[\rm (1)]   $\Phi^{(\mathrm{CHF})}(z;\alpha,\beta)$ ($\Phi^{(\mathrm{CHF})}(z)$, for short) is analytic in
  $\mathbb{C}\setminus \{\cup^8_{j=1}\Sigma_j\}$, where the contours $\Sigma_j$, $j=1,\ldots,8,$ are shown in Fig. \ref{fig:JCHF}.
\begin{figure}[h]
    \centering
    \includegraphics[width=0.5\textwidth]{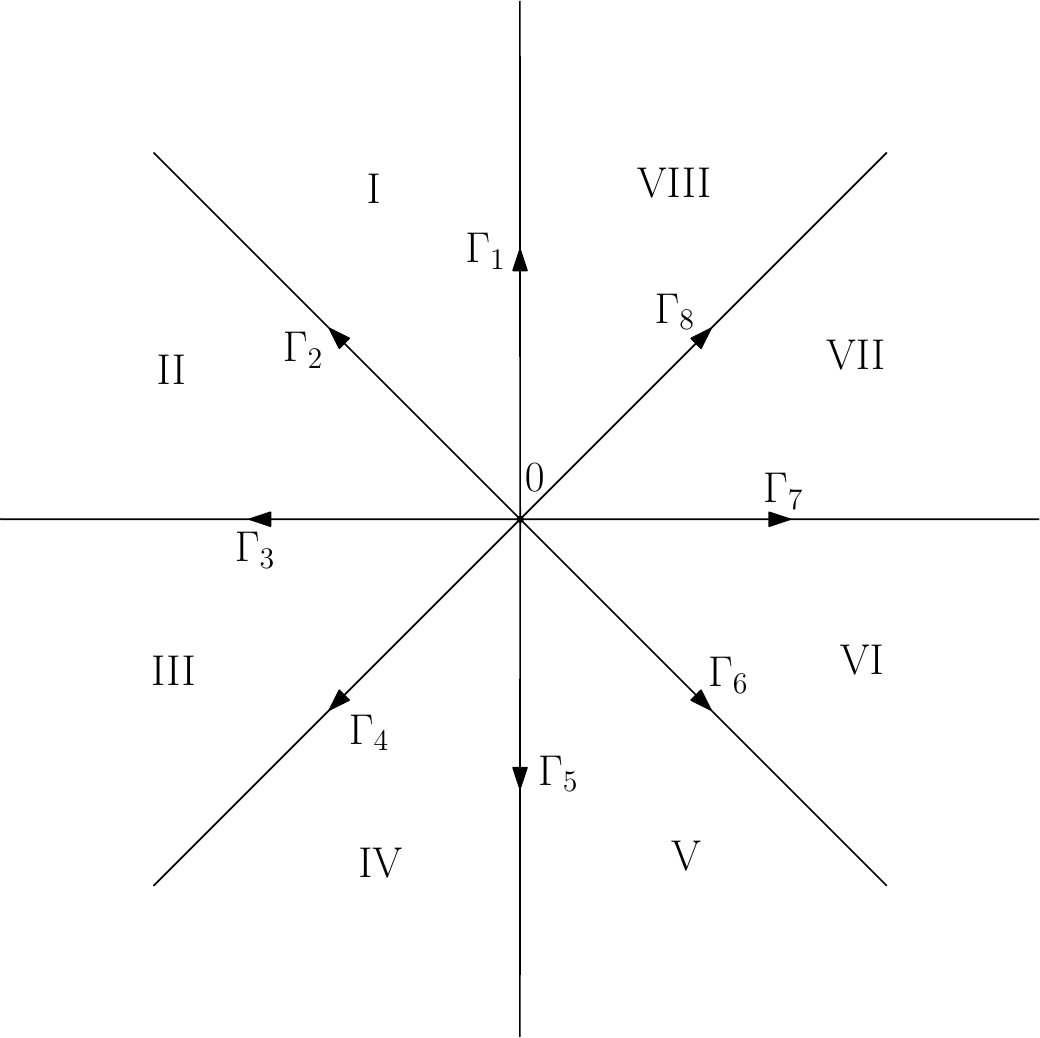}
    \caption{\small{The jump contours and regions of the RH problem for $\Phi^{(\mathrm{CHF})}(z)$.}}
    \label{fig:JCHF}
\end{figure}
  \item[\rm (2)] $\Phi^{(\mathrm{CHF})}(z)$ satisfies the  jump condition
  \begin{equation}\label{jumps-phi-c}
   \left(\Phi^{(\mathrm{CHF})}\right)_+(z)=\left(\Phi^{(\mathrm{CHF})}\right)_{-}(z) J_j(z), \quad z \in \Sigma_i,\quad j=1,\ldots,8,
  \end{equation}
  where
  \begin{equation*}
   J_1(z) = \begin{bmatrix}
    0 &   e^{-\pi i\beta} \\
    -  e^{\pi i\beta} &  0
    \end{bmatrix}, \quad J_2(z) = \begin{bmatrix}
    1 & 0 \\
    e^{ \pi i(\beta-2\alpha)} & 1
    \end{bmatrix},
                                                           \quad J_3(z) = J_7(z) = e^{\pi i\alpha\sigma_3},
  \end{equation*}
  \begin{equation*}
  J_4(z) = \begin{bmatrix}
    1 & 0 \\
    e^{ -\pi i(\beta-2\alpha)} & 1
    \end{bmatrix}, \quad
    J_5(z) = \begin{bmatrix}
    0 &   e^{\pi i\beta} \\
     -  e^{-\pi i\beta} &  0
     \end{bmatrix}, \quad
     J_6(z) = \begin{bmatrix}
     1 & 0 \\
     e^{- \pi i(\beta+2\alpha)} & 1
     \end{bmatrix},
  \end{equation*}
  and $$J_8(z) = \begin{bmatrix}
   1 & 0 \\
   e^{\pi i(\beta+2\alpha)} & 1
   \end{bmatrix}.
   $$

  \item[\rm (3)] As $z\to\infty$, $\Phi^{(\mathrm{CHF})}(z)$ satisfies the asymptotic behavior
 \begin{equation}\label{asym of phiCH}
    \begin{aligned}
        \Phi^{(\mathrm{CHF})}(z)=&\left(I+\frac{\alpha^2-\beta^2}{z}\left[\begin{array}{cc}
1 & \frac{\Gamma(\alpha-\beta)}{\Gamma(1+\alpha+\beta)} \\
-\frac{\Gamma(\alpha+\beta)}{\Gamma(1+\alpha-\beta)} & -1
\end{array}\right]+\ldots\right)z^{-\beta \sigma_3} e^{-\frac{1}{2} z \sigma_3}\\ & \begin{cases}e^{-\frac{1}{2} \alpha \pi i \sigma_3} e^{\beta \pi i \sigma_3}, &\frac{\pi}{2}<\arg z<\pi, \\
e^{\frac{1}{2} \alpha \pi i \sigma_3} e^{\beta \pi i \sigma_3}, & \pi <\arg z<\frac{3\pi}{2}, \\
{\left[\begin{array}{cc}
0 & -1 \\
1 & 0
\end{array}\right] e^{-\frac{1}{2} \alpha \pi i \sigma_3}}, &-\frac{\pi}{2}<\arg z< 0, \\[.4cm]
{\left[\begin{array}{cc}
0 & -1 \\
1 & 0
\end{array}\right] e^{\frac{1}{2} \alpha \pi i \sigma_3}}, & 0<\arg z<\frac{\pi}{2}.\end{cases}
    \end{aligned}
    \end{equation}

\item[\rm (4)] As $z \to 0$, $\Phi^{(\mathrm{CHF})}(z)$ has the following asymptotic behaviors
$$\Phi^{(\mathrm{CHF})}(z)=\begin{bmatrix}
                         O(|z|^{\alpha}) & O(|z|^{-|\alpha|}) \\
                         O(|z|^{\alpha}) & O(|z|^{-|\alpha|}) \end{bmatrix}, \quad \alpha> 0,$$
                         $$\Phi^{(\mathrm{CHF})}(z)=\begin{bmatrix}
                         O(1) & O(\ln |z|) \\
                         O(1) & O(\ln |z|)
\end{bmatrix}, \quad \alpha=0,$$
and
$$\Phi^{(\mathrm{CHF})}(z)=\begin{bmatrix}
                         O(|z|^{\alpha}) & O(|z|^{|\alpha|}) \\
                         O(|z|^{\alpha}) & O(|z|^{|\alpha|}) \end{bmatrix}, \quad \alpha< 0.$$
\end{description}

For the application in the present paper, we introduce a matrix-valued function $\Phi(z)$ defined via   $\Phi^{(\mathrm{CHF})}(z)$ by the
transformation:
\begin{equation}
\Phi(z)=e^{-\frac{1}{2}\beta\pi i\sigma_3}\Phi^{(\mathrm{CHF})}(e^{\frac{1}{2}\pi i}z ; \alpha, \beta)\left\{\begin{array}{ll}
e^{\frac{\alpha}{2} \pi i \sigma_3},&  \arg z \in\left(0, \frac{\pi}{2}\right), \\
e^{-\frac{\alpha}{2} \pi i \sigma_3},&\arg z\in \left(\frac{\pi}{2},\pi\right), \\
e^{\frac{\alpha+2 \beta}{2} \pi i \sigma_3}{\left[\begin{array}{cc}
0 & 1 \\
-1 & 0
\end{array}\right] },&\arg z\in \left(-\pi, -\frac{\pi}{2}\right),\\
 e^{-\frac{\alpha+2 \beta}{2} \pi i \sigma_3}{\left[\begin{array}{cc}
0 & 1 \\
-1 & 0
\end{array}\right]},&\arg z\in\left(-\frac{\pi}{2},0\right).
\end{array}\right.
\end{equation}

Then $\Phi(z)$ satisfies the following RH problem.
\begin{itemize}
    \item[\rm (1)]
  $\Phi(z)$ is   analytic in $\mathbb{C}\backslash \cup_{j=1}^5 \Sigma_{\Phi, j}$, where the oriented contours
    $$
    \begin{aligned}
& \Sigma_{\Phi, 1}=e^{\frac{\pi i}{4}} \mathbb{R}^{+}, \quad \Sigma_{\Phi, 2}=e^{\frac{3 \pi i}{4}} \mathbb{R}^{+}, \quad \Sigma_{\Phi, 3}=e^{-\frac{3 \pi i}{4}} \mathbb{R}^{+}, \\
& \Sigma_{\Phi, 4}=e^{-\frac{\pi i}{2}} \mathbb{R}^{+}, \quad \Sigma_{\Phi, 5}=e^{-\frac{\pi i}{4}} \mathbb{R}^{+},
\end{aligned}
    $$ are shown in Figure \ref{fig:CHF}.
\begin{figure}[h]
    \centering
    \includegraphics[width=0.5\textwidth]{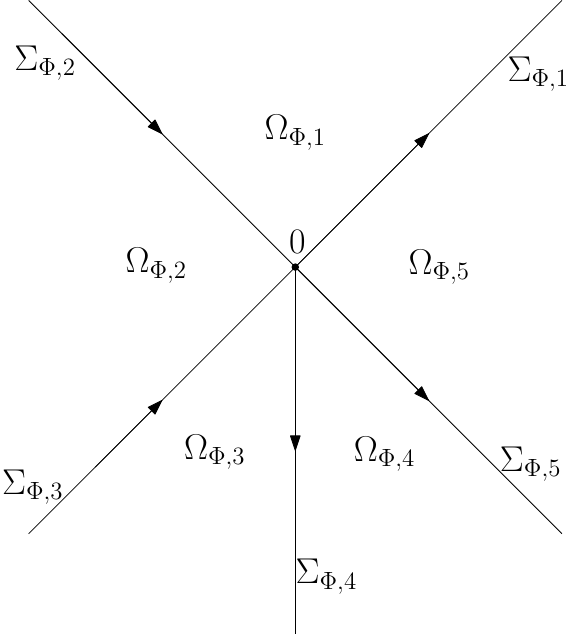}
    \caption{\small{The jump contours and regions of the RH problem for $\Phi(z)$.}}
    \label{fig:CHF}
\end{figure}

    \item[\rm (2)]  $\Phi(z)$ satisfies the  jump condition $\Phi_+(z)=\Phi_-(z)J_{\Phi}(z)$ on $\cup_{j=1}^5 \Sigma_{\Phi, j}$, where
    \begin{equation}\label{jump4phi}
     J_{\Phi}(z)=\begin{cases}\left[\begin{array}{cc}
1 & 0 \\
e^{-\pi i(\alpha-\beta)} & 1
\end{array}\right], & z \in \Sigma_{\Phi, 1}, \\[.4cm]
\left[\begin{array}{cc}
1 & 0 \\
e^{\pi i(\alpha-\beta)} & 1
\end{array}\right], & z \in \Sigma_{\Phi, 2}, \\[.4cm]
\left[\begin{array}{cc}
1 & -e^{-\pi i(\alpha-\beta)} \\
0 & 1
\end{array}\right], & z \in \Sigma_{\Phi, 3}, \\[.4cm]
e^{2 \pi i \beta \sigma_3}, & z \in \Sigma_{\Phi, 4}, \\
\left[\begin{array}{cc}
1 & -e^{\pi i(\alpha-\beta)} \\
0 & 1
\end{array}\right], & z \in \Sigma_{\Phi, 5} .\end{cases}
    \end{equation}

\item [\rm (3)] As $z\rightarrow\infty$,  we have
\begin{equation}\label{phi aym inf}
  \Phi(z)=\left(I+\frac{\Phi_1}{z}+\mathcal{O}\left(\frac{1}{z^2}\right)\right) z^{-\beta \sigma_3} e^{-\frac{i}{2} z \sigma_3},
\end{equation}
where $\Phi_1=\left[\begin{array}{cc}
-\left(\alpha^2-\beta^2\right) i & -e^{-\pi i \beta} \frac{\Gamma(1+\alpha-\beta)}{i \Gamma(\alpha+\beta)} \\
e^{\pi i \beta} \frac{\Gamma(1+\alpha+\beta)}{i \Gamma(\alpha-\beta)} & \left(\alpha^2-\beta^2\right) i
\end{array}\right]$, and arg $z\in(-\frac{\pi}{2},\frac{3\pi}{2}) $.

 \item [\rm (4)]
If  $\alpha>-\frac{1}{2}$ and $ 2\alpha \not\in \mathbb{N}$, we have as $z\to 0$
  \begin{equation}\label{eq: Phi0}
  \Phi(z)= \Phi^{(0)}(z)z^{\alpha \sigma_3} C_j,~~~z\to 0,~z\in \Omega_{\Phi, j}, ~j=1,2,3,4,5, \end{equation}
where $\Phi^{(0)}(z)$ is analytic near the origin, the regions $\Omega_{\Phi, j}$ are  shown in Figure \ref{fig:CHF}, and the branch for $z^{\alpha \sigma_3}$ is chosen such that $\arg z\in (-\pi/2, 3\pi/2)$.
The constant matrix
  $$C_1= \begin{bmatrix}
                                 1 &\frac{\sin(\pi(\alpha+\beta))}{\sin (2\pi \alpha)} \\
                                 0 &1
                                 \end{bmatrix},$$
and the constant matrices $C_j$, $j=2,3,4,5$ are determined by $C_1$ and the jump condition.
If $ 2\alpha \in \mathbb{N}$, we have as $z\to 0$
  \begin{equation}\label{eq: Phi02}
  \Phi(z)= \widehat{\Phi}^{(0)}(z)z^{\alpha \sigma_3} \begin{bmatrix}
                                 1 &(-1)^{2\alpha}\frac{\sin(\pi(\alpha+\beta))}{\pi} \ln  z \\
                                 0 &1
                                 \end{bmatrix} \widehat{C}_j, ~z\in \Omega_{\Phi, j}, ~j=1,2,3,4,5, \end{equation}
  where $\widehat{\Phi}^{(0)}(z)$ is analytic near the origin, the regions $\Omega_{\Phi, j}$ are shown in Figure \ref{fig:CHF}, and the branches for $z^{\alpha \sigma_3}$ and $\ln z$ are chosen such that $\arg z\in (-\pi/2, 3\pi/2)$,
  the constant matrix $\hat{C}_1$ is the identity matrix and the other constant matrices are determined by $\hat{C}_1$ and the jump condition.

\end{itemize}

The solution to the  RH problem can be constructed explicitly by using the confluent hypergeometric function
\begin{equation}\label{solution4phi}
   \begin{aligned}
       \Phi(z)
    = &  e^{-\frac{i z}{2}}\left[\begin{array}{cc}
e^{-\frac{\pi i(\alpha+\beta)}{2}} \frac{\Gamma(1+\alpha-\beta)}{\Gamma(1+2 \alpha)} \phi(\alpha+\beta, 1+2 \alpha, i z) & -e^{\frac{\pi i(\alpha-\beta)}{2}} \frac{\Gamma(2 \alpha)}{\Gamma(\alpha+\beta)} \phi(-\alpha+\beta, 1-2 \alpha, i z) \\
e^{-\frac{\pi i(\alpha-\beta)}{2}} \frac{\Gamma(1+\alpha+\beta)}{\Gamma(1+2 \alpha)} \phi(1+\alpha+\beta, 1+2 \alpha, i z) & e^{\frac{\pi i(\alpha+\beta)}{2}} \frac{\Gamma(2 \alpha)}{\Gamma(\alpha-\beta)} \phi(1-\alpha+\beta, 1-2 \alpha, i z)
\end{array}\right]\\
&z^{\alpha \sigma_3} \left[\begin{array}{cc}
1 & \frac{\sin (\pi(\alpha+\beta))}{\sin (2 \pi \alpha)} \\
0 & 1
\end{array}\right]  \end{aligned}
\end{equation}
for  $z\in \Omega_{\Phi,1}$,  $\alpha>-1/2$ and $2 \alpha \notin \mathbb{N}$; see \cite{cik, ik}. Here, $\phi(a,b,z)$ is the confluent hypergeometric function defined in \eqref{eq:chf}. The expression of   $\Phi(z) $ in the other regions is then determined by using \eqref{solution4phi} and the jump condition \eqref{jump4phi}. The case $2\alpha\in \mathbb{N}$ can be constructed in a similar way and  particularly the first column of  $\Phi(z)$ is given by the same formula in \eqref{solution4phi}.

\end{appendices}

\section*{Acknowledgements}
The work of Shuai-Xia Xu was supported in part by the National Natural Science Foundation of China under grant numbers 11971492 and 12371257, and by Guangdong Basic and Applied Basic Research Foundation (Grant No. 2022B1515020063). Yu-Qiu Zhao was supported in part by the National Natural Science Foundation of China under grant numbers  11971489 and
12371077.

\end{document}